\documentclass[a4paper,runningheads]{llncs}
\pdfoutput=1
\usepackage[utf8]{inputenc}
\usepackage[T1]{fontenc}
\usepackage[ngerman,english]{babel}
\usepackage{llncsstuffthm}
\usepackage{cwdefs}
\usepackage{pgfkeys}
\usepackage{plabels}
\usepackage{amssymb}
\usepackage{amsfonts}
\usepackage{multirow}
\usepackage{longtable}
\usepackage{booktabs}
\usepackage{amsmath}
\usepackage{bigdelim}
\usepackage{comment}
\usepackage{xspace}
\usepackage{booktabs}
\usepackage{graphicx}
\usepackage{setspace}
\usepackage{upgreek}
\usepackage{enumitem}
\usepackage{rotating}
\usepackage{tikz}
\usetikzlibrary{shapes.geometric}
\usetikzlibrary{shapes.symbols}
\usetikzlibrary{positioning}
\usepackage{array}
\usepackage[titles]{tocloft}
\usepackage{xcolor}
\usepackage{colortbl}
\usepackage{capt-of}
\usepackage{pdfpages}
\usepackage{wrapfig}
\usepackage{fancyvrb}
\usepackage[hidelinks]{hyperref}

\usepackage{orcidlink}\renewcommand{\orcidID}[1]{\orcidlink{#1}}

\newcommand{\symset}[1]{\mathcal{#1}}

\renewcommand{\us}{\symset{U}}
\renewcommand{\vs}{\symset{V}}
\newcommand{\fgs}{\symset{F\!G}}
\newcommand{\ffs}{\symset{F}}
\newcommand{\ggs}{\symset{G}}
\newcommand{\XS}{\symset{X}}
\newcommand{\YS}{\symset{Y}}
\newcommand{\ZS}{\symset{Z}}
\newcommand{\SSS}{\symset{S}}

\newcommand{\nlit}[1]{\f{lit}(#1)}
\newcommand{\nclause}[1]{\f{clause}(#1)}

\newcommand{\nside}[1]{\f{side}(#1)}
\newcommand{\ntgt}[1]{\f{tgt}(#1)}
\newcommand{\nipol}[1]{\f{ipol}(#1)}
\newcommand{\npath}[2]{\f{path}_{#1}(#2)}
\newcommand{\ncopy}[1]{\f{copy}(N)}

\newcommand{\cnf}[1]{\f{cnf}(#1)}
\newcommand{\dnf}[1]{\f{dnf}(#1)}

\newcommand{\INVC}[1]{\f{INV}_{\f{C}}(#1)}
\newcommand{\INVD}[1]{\f{INV}_{\f{D}}(#1)}
\newcommand{\INVX}[1]{\f{INV}_{\f{X}}(#1)}

\newcommand{\MC}{M_{\f{C}}}
\newcommand{\MD}{M_{\f{D}}}

\newcommand{\undefined}{\bot}

\newcommand{\dual}[1]{\f{dual}(#1)}

\newcommand{\prenum}{1.0em}

\newcommand{\univ}{U\xspace}

\newcommand{\vgtrred}{VGT-range-restricted\xspace}
\newcommand{\vgtrrion}{VGT-range-restriction\xspace}

\newcommand{\unirred}{\univ-range-restricted\xspace}

\tikzset{
itria/.style={
  draw,
  solid,
  thin,
  isosceles triangle,
  isosceles triangle apex angle=60,
  shape border rotate=90,yshift=-6.8ex}
}

\newcommand{\tableauscale}{0.7}

\newcommand{\extabld}{10ex}

\newcommand{\rewrite}{\Rightarrow}
\newcommand{\rewritereg}{\stackrel{\!\!\scriptsize{simp}\!\!\raisebox{-2pt}{\rule{0pt}{1pt}}}{\Rightarrow}}

\newcommand{\algoinput}{\smallskip \noindent\textsc{Input: }}
\newcommand{\algooutput}{\smallskip \noindent\textsc{Output: }}
\newcommand{\algomethod}{\smallskip \noindent\textsc{Method: }}
\newcommand{\algoskip}{\vspace{2pt}}

\newcommand{\nbadlits}[1]{\f{bad\hyph literals}(#1)}
\newcommand{\ncode}[1]{\f{path\hyph string}(#1)}

\definecolor{tcolbbbbg}{rgb}{0.8,0.8,0.8}
\definecolor{tcolaaabg}{rgb}{1.0,1.0,1.0}
\newcommand{\taaa}[1]{\colorbox{tcolaaabg}{$#1\vbar$}}
\newcommand{\tbbb}[1]{\colorbox{tcolbbbbg}{$#1\vbar$}}

\newcommand{\vbar}{\raisebox{-0.60ex}{\rule{0pt}{2.35ex}}}

\renewcommand{\taaa}[1]{\colorbox{tcolaaabg}{$#1\vbar$}}
\renewcommand{\tbbb}[1]{\colorbox{tcolbbbbg}{$#1\vbar$}}

\newcommand{\tbbbtxt}[1]{\colorbox{tcolbbbbg}{#1}}

\newcommand{\nhphantom}[1]{\sbox0{#1}\hspace*{-\the\wd0}}
\newcommand{\nannot}[1]%
           {\hspace{0.2em}{[}#1{]}\nhphantom{\hspace{0.2em}{[}#1{]}}}
\newcommand{\nannotw}[1]%
           {\hspace{0.6em}{[}#1{]}\nhphantom{\hspace{0.6em}{[}#1{]}}}

\newcommand{\nannotmm}[1]%
           {\hspace{0.2em}\underline{{[}#1{]}}\nhphantom{\hspace{0.2em}\underline{{[}#1{]}}}}
\newcommand{\nannotwmm}[1]%
           {\hspace{0.6em}\underline{{[}#1{]}}\nhphantom{\hspace{0.6em}\underline{{[}#1{]}}}}

\newcommand{\m}[1]{\mathit{#1}}

\newcommand{\varfun}{\m{\mathcal{V}\hspace{-0.11em}ar}}

\newcommand{\var}[1]{\varfun(#1)}

\newcommand{\vall}[1]{\var{#1}}
\newcommand{\vpos}[1]{\varfun^{+}(#1)}
\newcommand{\vneg}[1]{\varfun^{-}(#1)}

\renewcommand{\VV}{\mathcal{V}}

\newcommand{\tmaxfun}{\hyph\m{\mathcal{M}\hspace{-0.11em}ax}}
\newcommand{\vtmax}[1]{\VV\tmaxfun({#1})}
\newcommand{\vtmaxpos}[1]{\VV\tmaxfun^{+}(#1)}
\newcommand{\vtmaxneg}[1]{\VV\tmaxfun^{-}(#1)}

\newcommand{\xtmax}[1]{\tmaxfun({#1})}
\newcommand{\xtmaxpos}[1]{\tmaxfun^{+}(#1)}
\newcommand{\xtmaxneg}[1]{\tmaxfun^{-}(#1)}

\newcommand{\voc}[1]{\m{\mathcal{V}\hspace{-0.11em}oc}^\pm{(#1)}}

\newcommand{\VG}{\mathcal{U}}
\newcommand{\VF}{\mathcal{E}}
\newcommand{\VU}{\mathcal{U}}
\newcommand{\VE}{\mathcal{E}}
\newcommand{\VX}{\mathcal{X}}

\newcommand{\VC}{\mathcal{C}}

\newcommand{\pred}[1]{\m{\mathcal{P}\hspace{-0.11em}red}^\pm(#1)}

\newcommand{\fun}[1]{\m{\mathcal{F}\hspace{-0.18em}un}(#1)}

\newcommand{\aaa}{\red{\f{red}}}
\newcommand{\bbb}{\blue{\f{blue}}}
\renewcommand{\aaa}{\f{F}}
\renewcommand{\bbb}{\f{G}}
\newcommand{\LL}{\f{F}}
\newcommand{\RR}{\f{G}}
\newcommand{\sided}{two-sided\xspace}

\newcommand{\FL}{F}
\renewcommand{\GR}{G}

\newcommand{\CTIF}{CTIF\xspace}

\newcommand{\HG}{H_{\textsc{grd}}}

\newcounter{ibcounter}

\newcommand{\sterms}{\text{-terms}}

\newcommand{\Vampire}{\textsf{Vampire}\xspace}

\newcommand{\EProver}{\textsf{E}\xspace}
\newcommand{\Princess}{\textsf{Princess}\xspace}

\newcommand{\ProverN}{\textsf{Prover9}\xspace}
\newcommand{\leanCoP}{\textsf{leanCoP}\xspace}

\newcommand{\CMProver}{\textsf{CMProver}\xspace}
\newcommand{\SETHEO}{\textsf{SETHEO}\xspace}

\newcommand{\PIE}{\textsf{PIE}\xspace}
\newcommand{\Prooftrans}{\textsf{Prooftrans}\xspace}

\newcommand{\FC}{F'}
\newcommand{\NGC}{G'}

\newcommand{\ipctx}{\la F, G, \FC, \NGC, \ffs, \ggs, \VF, \VG, \VC, \VV\ra}
\newcommand{\ipctxprefix}{\la F, G, \FC, \NGC\ra}

\newcommand{\supplementarynote}{Proofs of nontrivial claims that are not
proven in the body of the paper are supplemented in the appendix}
 
\newcommand{\appref}[1]{App.~\ref{#1}}
\begin{document}

\title{Range-Restricted and Horn Interpolation through Clausal
  Tableaux\thanks{Funded by the Deutsche Forschungsgemeinschaft (DFG, German
    Research Foundation) -- Project-ID~457292495. The work was supported by
    the North-German Supercomputing Alliance (HLRN).}}

\titlerunning{Range-Restricted Interpolation through Clausal Tableaux}

\author{Christoph Wernhard\orcidID{0000-0002-0438-8829}}
\authorrunning{C. Wernhard}

\institute{University of Potsdam
  \email{info@christophwernhard.com}}

\maketitle

\begin{abstract}
We show how variations of range-restriction and also the Horn property can be
passed from inputs to outputs of Craig interpolation in first-order logic. The
proof system is clausal tableaux, which stems from first-order ATP. Our
results are induced by a restriction of the clausal tableau structure, which
can be achieved in general by a proof transformation, also if the source proof
is by resolution/paramodulation. Primarily addressed applications are query
synthesis and reformulation with interpolation. Our methodical approach
combines operations on proof structures with the immediate perspective of
feasible implementation through incorporating highly optimized first-order
provers.
\end{abstract}

\section{Introduction}

We show how variations of range-restriction and also the Horn property can be
passed from inputs to outputs of Craig interpolation in first-order logic.
The primarily envisaged application field is synthesis and reformulation of
queries with interpolation \cite{nash:2010,toman:wedell:book,benedikt:book}.
Basically, the sought target query~$R$ is understood there as the right side
of a definition of a given query~$Q$ within a given background knowledge
base~$K$, i.e., it holds that $K \entails (Q \equi R)$, where the vocabulary
of $R$ is in a given set of permitted target symbols.
In first-order logic, the formulas~$R$ can be characterized as the Craig
interpolants of $K \land Q$ and $\lnot K' \lor Q'$, where $K,Q$ are copies of
$K',Q'$ with the symbols not allowed in $R$ replaced by fresh symbols
\cite{craig:uses}. Formulas $R$ exist if and only if the entailment $K \land Q
\entails \lnot K' \lor Q'$ holds. They can be constructed as Craig
interpolants from given proofs of the entailment in a suitable calculus.

In databases and knowledge representation, syntactic fragments of first-order
logic ensure desirable properties, for example domain independence. Typically,
for given $K$ and $Q$ in some such fragment, also $R$ must be in some specific
fragment to be usable as a query or as a knowledge base component. Our work
addresses this by showing for certain such fragments how membership is passed
on to interpolants and thus to the constructed right sides of definitions.
The fragment in focus here is a variant of range-restriction from
\cite{vgt}, known as a rather general syntactic condition to ensure domain
independence \cite[p.~97]{foundations}. It permits conversion into a shape
suitable for ``evaluation'' by binding free and quantified variables
successively to the members of given predicate extensions. Correspondingly, if
the vocabulary is relational, a range-restricted formula can be translated
into a relational algebra expression. First-order representations of
widely-used classes of integrity constraints, such as tuple-generating
dependencies, are sentences that are range-restricted in the considered sense.

As proof system we use \name{clausal tableaux}
\cite{letz:clausal:deduktion,letz:habil,handbook:tableaux:letz,handbook:ar:haehnle,letz:stenz:handbook},
devised in the 1990s to take account of automated first-order provers that may
be viewed as enumerating tree-shaped proof structures, labeled with instances
of input clauses.\footnote{Alternate accounts and views are provided by model
  elimination \cite{loveland:1978} and the connection method
  \cite{bibel:atp:1982,bibel:otten:2020}.} Such systems include the Prolog
Technology Theorem Prover \cite{pttp}, \SETHEO \cite{setheo:92}, \mbox{\leanCoP}
\cite{otten:2003:leancopinabstract,leancop} and \CMProver
\cite{cw-mathlib,cw:pie:2016,cw:pie:2020,rwzb:lemmas:2023}. As shown in
\cite{cw:ipol}, a \emph{given} closed clausal tableau is quite well-suited as
a proof structure to extract a Craig interpolant. Via the translation of a
resolution deduction tree \cite{chang:lee} to a clausal tableau in cut normal
form \cite{letz:habil,cw:ipol} this transfers also to interpolation from a
given re\-so\-lution/para\-mo\-dulation proof.

Since the considered notion of range-restriction is based on prenexing and
properties of both a CNF and a DNF representation of the formula, it fits
well with the common first-order ATP setting involving Skolemization and
clausification and the ATP-oriented interpolation on the basis of clausal
tableaux, where in a first stage the propositional structure of the
interpolant is constructed and in a second stage the quantifier prefix.

Our strengthenings of Craig interpolation are induced by a specific
restriction of the clausal tableau structure, which we call \name{hyper},
since it relates to the proof structure restrictions of hyperresolution
\cite{robinson:hyper} and hypertableaux \cite{hypertableaux}. However, it is
considered here for tree structures with rigid variables. A proof
transformation that converts an arbitrary closed clausal tableau to one with
the hyper property shows that the restriction is w.l.o.g. and, moreover,
allows the prover unhampered search for the closed clausal tableaux or
resolution/paramodulation proof underlying interpolation.

\paragraph{Structure of the Paper.}

Section~\ref{sec-prelim} summarizes preliminaries, in particular interpolation
with clausal tableaux \cite{cw:ipol}. Our main result on strengthenings of
Craig interpolation for range-restricted formulas is developed in
Sect.~\ref{sec-ipol-rr}. Section~\ref{sec-horn} discusses Craig interpolation
from a Horn formula, also combined with range-restriction. The proof
transformation underlying these results is introduced in
Sect.~\ref{sec-hyper-convert}. We conclude in Sect.~\ref{sec-conclusion} with
discussing related work, open issues and perspectives.

\supplementarynote.
An implementation with the
\PIE environment
\cite{cw:pie:2016,cw:pie:2020}\footnote{http://cs.christophwernhard.com/pie}
is in progress.

\section{Notation and Preliminaries}
\label{sec-prelim}

\subsection{Notation}

We consider \defname{formulas} of first-order logic. An \defname{NNF formula}
is a quantifier-free formula built up from \defname{literals} (atoms or
negated atoms), truth-value constants $\true, \false$, conjunction and
disjunction. A \defname{CNF formula}, also called \defname{clausal formula},
is an NNF formula that is a conjunction of disjunctions (\defname{clauses}) of
literals. A \defname{DNF formula} is an NNF formula that is a disjunction of
conjunctions (\defname{conjunctive clauses}) of literals.
The complement of a literal~$L$ is denoted by $\du{L}$.
An occurrence of a subformula in a formula has positive (negative)
\defname{polarity}, depending on whether it is in the scope of an even (odd)
number of possibly implicit occurrences of negation.
Let $F$ be a formula. $\var{F}$ is set of its free variables. $\vpos{F}$
($\vneg{F}$) is the set of its free variables with an occurrence in an atom
with positive (negative) polarity. $\fun{F}$ is the set of functions occurring
in it, including constants, regarded here throughout as 0-ary functions.
$\pred{F}$ is the set of pairs $\la p, \mathit{pol}\ra$, where $p$ is a
predicate and $\mathit{pol} \in \{{+},{-}\}$, such that an atom with
predicate~$p$ occurs in $F$ with the polarity indicated by $\mathit{pol}$.
$\voc{F}$ is $\fun{F} \cup \pred{F}$. A \defname{sentence} is a formula
without free variables. An NNF is \defname{ground} if it has no variables. If
$S$ is a set of terms, we call its members $S$-terms. The $\entails$ symbol
expresses semantic entailment.

\subsection{Clausal First-Order Tableaux}
\label{sec-tableaux}

A \defname{clausal tableau} (briefly \name{tableau}) \defname{for} a clausal
formula $F$ is a finite ordered tree whose nodes~$N$ with exception of the
root are labeled with a literal $\nlit{N}$, such that for each node~$N$ the
disjunction of the literals of all its children in their left-to-right order,
$\nclause{N}$, is an instance of a clause in~$F$.
A branch of a tableau is \defname{closed} iff it contains nodes with
complementary literals. A node is \defname{closed} iff all branches through it
are closed. A tableau is \defname{closed} iff its root is closed.
A node is \defname{closing} iff it has an ancestor with complementary literal.
With a closing node~$N$, a particular such ancestor is associated as
\defname{target of}~$N$, written~$\ntgt{N}$.
A tableau is \defname{regular} iff no node has an ancestor with the same
literal and is \defname{leaf-closing} iff all closing nodes are leaves.
A closed tableau that is leaf-closing is called \defname{leaf-closed}. Tableau
simplification can convert any tableau to a regular and leaf-closing tableau
for the same clausal formula, closed if the original tableau is so.
Regularity is achieved by repeating the following operation
\cite[Sect.~2.1.3]{letz:habil}: Select a node~$N$ with an ancestor that has
the same literal, remove the edges originating in the parent of $N$ and
replace them with the edges originating in $N$. The leaf-closing property is
achieved by repeatedly selecting an inner node $N$ that is closing and
removing the edges originating in $N$.
All occurrences of variables in (the literal labels of) a tableau are free and
their scope spans the whole tableau. That is, we consider
\name{free-variable tableaux} \cite[p.~158ff]{handbook:tableaux:letz} with
\name{rigid} variables \cite[p.~114]{handbook:ar:haehnle}. A tableau without
variables is called \defname{ground}.
The universal closure of a clausal formula~$F$ is unsatisfiable iff there
exists a closed clausal tableau for~$F$. This holds also if \name{clausal
  tableau} is restricted by the properties \name{ground}, \name{regular} and
\name{leaf-closing} in arbitrary combinations.

\subsection{Interpolation with Clausal Tableaux}
\label{sec-ipol-basic}

Craig's interpolation theorem \cite{craig:linear,craig:2008:road} along with
Lyndon's observation on the preservation of predicate polarities \cite{lyndon}
ensures for first-order logic the existence of \name{Craig-Lyndon
  interpolants}, defined as follows.
Let $F, G$ be formulas such that $F \entails G$.  A \defname{Craig-Lyndon
  interpolant} of $F$ and $G$ is a formula $H$ such that
(1)~$F \entails H$ and $H \entails G$.
(2)~$\voc{H} \subseteq \voc{F} \cap \voc{G}$.
(3)~$\var{H} \subseteq \var{F} \cap \var{G}$.
The perspective of validating an entailment $F \entails G$ by showing
unsatisfiability of $F \land \lnot G$ is reflected in the notion of
\defname{reverse Craig-Lyndon interpolant} of $F$ and $\GR$, defined as
Craig-Lyndon interpolant of $F$ and $\lnot \GR$.

\begin{wrapfigure}[16]{R}{5.3cm}
  \centering
  \vspace{-0.8cm}
  \hspace{-0.2cm}
  \begin{tikzpicture}[scale=0.83,
      sibling distance=9em,level distance=\extabld,
      every node/.style = {transform shape,anchor=mid}]]
      \node (a) {\vbar\textbullet}
      child { node {$\tbbb{\lnot \fr(\fa)}$\nannot{$\fq(\fa)$}}
        child { node {$\tbbb{\lnot \fq(\fa)}$\nannot{$\fq(\fa)$}}
          child { node {$\taaa{\lnot \fp(\fa)}$\nannot{$\false$}}
            child { node {$\taaa{\fp(\fa)}$\nannot{$\false$}} }
          }
          child { node {$\taaa{\fq(\fa)}$\nannot{$\fq(\fa)$}} }
        }
        child { node {$\tbbb{\fr(\fa)}$\nannot{$\true$}}
        }
      };
  \end{tikzpicture}
  \vspace{-2pt}
  \caption{A two-sided clausal tableau.}
  \label{fig-tab-egr}
\end{wrapfigure}

Following \cite{cw:ipol}, our interpolant construction is based on a
generalization of clausal tableaux where nodes have an additional \name{side}
label that is shared by siblings and indicates whether the tableau clause is
an instance of an input clause derived from the formula $F$ or of the formula
$G$ of the statement $F \land \GR \entails \false$ underlying the reverse
interpolant. Thus, a \defname{\sided clausal tableau } \defname{for} clausal
formulas $\FL$ \defname{and} $\GR$ is a tableau for $\FL \land \GR$ whose
nodes~$N$ with exception of the root are labeled additionally with a
\name{side} $\nside{N} \in \{\aaa, \bbb\}$, such that (1)~if $N$ and
$N^\prime$ are siblings, then $\nside{N} = \nside{N^\prime}$; (2)~if $N$ has a
child $N'$ with $\nside{N^\prime} = \aaa$, then $\nclause{N}$ is an instance
of a clause in $\FL$, and if~$N$ has a child $N'$ with $\nside{N^\prime} =
\bbb$, then $\nclause{N}$ is an instance of a clause in $\GR$. We also refer
to the side of the children of a node~$N$ as \defname{side of} $\nclause{N}$.
For $\mathit{side} \in \{\aaa,\bbb\}$ define $\npath{\mathit{side}}{N}\;
\eqdef \bigwedge_{N^\prime \in \mathit{Path} \text{ and } \nside{N^\prime} =
  \mathit{side}} \nlit{N^\prime}$, where $\mathit{Path}$ is the union of the
set of the ancestors of $N$ and $\{N\}$.

Let $N$ be a node of a leaf-closed \sided clausal tableau. The value of
\defname{$\nipol{N}$} is an NNF formula, defined inductively as specified with
the tables below, the left for the base case where $N$ is a leaf, the right
for the case where $N$ is an inner node with children $N_1, \ldots, N_n$.

\medskip

\noindent
\begin{minipage}[t]{0.5\textwidth}\small
  \centering
$\begin{array}{c@{\hspace{1em}}c@{\hspace{1em}}c}
\nside{N} &  \nside{\ntgt{N}} & \nipol{N}\\\midrule
\aaa & \aaa & \false\\[0.5ex]
\aaa & \bbb & \nlit{N}\\[0.5ex]
\bbb & \aaa & \du{\nlit{N}}\\[0.5ex]
\bbb & \bbb & \true
\end{array}$
\end{minipage}%
\begin{minipage}[t]{0.5\textwidth}\small
  \centering
  $\begin{array}{c@{\hspace{1em}}c}
    \nside{N_1} & \nipol{N}\\\midrule
    \aaa & \bigvee_{i=1}^n \nipol{N_i}\\[1ex]
    \bbb & \bigwedge_{i=1}^n \nipol{N_i}
  \end{array}$
\end{minipage}

\begin{examp}
  \label{examp-ground}
   Figure~\ref{fig-tab-egr} shows a \sided tableau for $F = \fp(\fa) \land
   (\lnot \fp(\fa) \lor \fq(\fa))$ and $G = (\lnot \fq(\fa) \lor \fr(\fa))
   \land \lnot \fr(\fa)$. \tbbbtxt{Side $\bbb$} is indicated by gray
   background. For each node the value of $\f{ipol}$, after truth-value
   simplification, is annotated in brackets. The clauses of the tableau are
   $\lnot \fr(\fa)$ and $\lnot \fq(\fa) \lor \fr(\fa)$, which have side
   $\bbb$, and $\lnot \fp(\fa) \lor \fq(\fa)$ and $\fp(\fa)$, which have side
   $\aaa$. If $N$ is the node shown bottom left, labeled with $\fp(\fa)$, then
   $\npath{\aaa}{N} = \lnot \fp(\fa) \land \fp(\fa)$ and $\npath{\bbb}{N} =
   \lnot \fr(\fa) \land \lnot \fq(\fa)$.
\end{examp}

\begin{figure}[p]
\centering\small

\fbox{
\begin{minipage}{0.96\textwidth}
 \algoinput First-order formulas $F$ and $G$ such that $F \entails G$.

 \algomethod
 \begin{enumerate}
 \item \label{step-const} \name{Free variables to placeholder constants.} Let
   $F_c$ and $G_c$ be the sentences obtained from $F$ and $G$ by replacing
   each free variable with a dedicated fresh constant.
 \item \label{step-pre} \name{Skolemization and clausification.} Apply there
   conversion to prenex form and second-order Skolemization independently to
   $F_c$ and to $\lnot G_c$, resulting in disjoint sets of fresh Skolem
   functions $\ffs', \ggs'$, clausal formulas $F',G'$, and sets $\us' =
   \var{F'}, \vs' = \var{G'}$ of variables such that
   \[\begin{array}{r@{\hspace{1em}}l}
   \text{(a)} & F_c \equiv \exists \ffs' \forall \us' F'
   \text{ and } \lnot G_c \equiv \exists \ggs' \forall \vs' G'.\\
   \text{(b)} & \voc{F'} \subseteq \voc{F_c} \cup \ffs' \text{ and }
   \voc{\lnot G'} \subseteq \voc{G_c} \cup \ggs'.\\
   \text{(c)} & \forall \us' \forall \vs' (F' \land G') \entails \false.\\
   \end{array}
   \]
   In case $F'$ or $G'$ contains the empty clause, exit with result $H \eqdef
   \false$ or $H \eqdef \true$, respectively.
 \item \label{step-tab} \name{Tableau computation.} Compute a leaf-closed
   clausal tableau for the clausal formula $F' \land G'$. This can be
   obtained, for example, from a clausal tableaux prover for clausal
   first-order formulas.

 \item \label{step-grounding} \name{Tableau grounding.} Instantiate all
   variables of the tableau with ground terms built up from functions in $F'
   \land G'$ and possibly also fresh functions $\SSS = \SSS_1 \uplus \SSS_2$.
   Observe that the grounded tableau is still a leaf-closed tableau for $F'
   \land G'$.

 \item \label{step-side-assignment} \name{Side assignment.} Convert the ground
   tableau to a \sided tableau for $F'$ and~$G'$ by attaching appropriate
   \name{side} labels to all nodes except the root. This is always possible
   because every clause of the tableau is an instance of a clause in $F'$ or
   in $G'$.

 \item \label{step-extract} \name{Ground interpolant extraction.} Let $\HG$ be
   the value of $\nipol{N_0}$, where $N_0$ is the root of the tableau.

 \item \label{step-lifting} \name{Interpolant lifting.} Let $\ffs \eqdef \ffs'
   \cup (\fun{F} \setminus \fun{G}) \cup \SSS_1$ and let $\ggs \eqdef \ggs'
   \cup (\fun{G} \setminus \fun{F}) \cup \SSS_2$. Let $\fgs$ stand for $\ffs
   \cup \ggs$. An \defname{$\fgs$-maximal occurrence} of an $\fgs$-term in a
   formula is an occurrence that is not within another $\fgs$-term. Let
   $\{t_1, \ldots, t_n\}$ be the set of the $\fgs\sterms$ with an
   $\fgs$-maximal occurrence in $\HG$, ordered such that if $t_i$ is a subterm
   of $t_j$, then $i < j$. Let $\{v_1,\ldots,v_n\}$ be a set of fresh
   variables. For $i \in \{1,\ldots,n\}$ define the quantifiers $Q_i$ as
   $\exists$ if $t_i \in \ffs\sterms$ and as $\forall$ if $t_i \in
   \ggs\sterms$. Let \[H_c \eqdef Q_1 v_1 \ldots Q_n v_n\, \HG',\] where
   $\HG'$ is obtained from $\HG$ by replacing all $\fgs$-maximal occurrences
   of terms $t_i$ with variable~$v_i$, simultaneously for all $i \in
   \{1,\ldots,n\}$.

 \item \label{step-constants-to-vars} \name{Placeholder constants to free
   variables.} Let $H$ be $H_c$ after replacing any constants that were
   introduced in step~\ref{step-const} with their corresponding variables.

 \end{enumerate}

 \algooutput Return $H$, a Craig-Lyndon interpolant of the input formulas $F$
 and~$G$.
 \end{minipage}
}

 \caption{The \CTIF Procedure for Craig-Lyndon Interpolation
   \cite{cw:ipol}.}
 \label{fig-ctif}
\end{figure}

If $N_0$ is the root of a \sided tableaux for clausal \emph{ground} formulas
$F$ and $G$, then $\nipol{N_0}$ is a Craig-Lyndon interpolant of $F$ and
$\lnot G$.\footnote{So far, the interpolation method is a variation
of well-known methods
for sequent systems \cite{smullyan:book:68,takeuti:1987}
and analytic tableaux \cite{fitting:book} when
restricted to propositional formulas.}
The \CTIF (\name{Clausal Tableau Interpolation for First-Order
  Formulas}) procedure (Fig.~\ref{fig-ctif}) \cite{cw:ipol} extends this to a
two-stage \cite{bonacina:15:on,huang:95} (inductive construction and lifting)
interpolation method for full first-order logic. It is complete (yields a
Craig-Lyndon interpolant for all first order formulas~$F$ and $G$ such that $F
\entails G$) under the assumption that the method for tableau computation in
Step~\ref{step-tab} is complete (yields a closed tableau for all unsatisfiable
clausal formulas). Some steps leave room for inter\-polation-specific
heuristics: In step~\ref{step-grounding} the choice of the terms used for
grounding; in step~\ref{step-side-assignment} the choice of the side assigned
to clauses that are an instance of both a clause in $F'$ and a clause in
$G'$; and in step~\ref{step-lifting} the quantifier prefix, which is
constrained just by a partial order.

\begin{examp}
  Let $F \eqdef \forall x\, \fp(x) \land \forall x\, (\lnot \fp(x) \lor
  \fq(x))$ and let $G \eqdef \forall x\, (\lnot \fq(x) \lor \fr(x)) \imp
  \fr(\fa)$. Clausifying $F$ and $\lnot G$ then yields $F' = \fp(x) \land (\lnot
  \fp(x) \lor \fq(x))$ and $G' = (\lnot \fq(x) \lor \fr(x)) \land \lnot
  \fr(\fa)$. The tableau from Fig.~\ref{fig-tab-egr} is a leaf-closed ground
  tableau for $F'$ and $G'$ and we obtain $\fq(\fa)$ as $\HG$. Lifting for
  $\ffs = \{\}$ and $\ggs = \{\fa\}$ yields the interpolant $H = \forall v_1\,
  \fq(v_1)$.
\end{examp}

\begin{examp}
  Let $F \eqdef \forall x \forall y\, \fp(x,\ff(x),y)$ and let $G \eqdef
  \exists x \fp(\fa,x,\fg(x))$. Clausifying yields $F' = \fp(x,\ff(x),y)$ and
  $G' = \lnot \fp(\fa,z,\fg(z))$. We obtain $\fp(\fa,\ff(\fa),\fg(\ff(\fa)))$
  as $\HG$. Lifting is for $\ffs = \{\ff\}$ and $\ggs = \{\fa, \fg\}$ with
  $t_1 = \fa$, $t_2 = \ff(\fa)$ and $t_3 = \fg(\ff(\fa))$. It yields $H =
  \forall v_1 \exists v_2 \forall v_3\, \fp(v_1, v_2, v_3)$.
\end{examp}

\section{Interpolation and Range-Restriction}
\label{sec-ipol-rr}

We now develop our main result on strengthenings of Craig interpolation for
range-restricted formulas.

\subsection{CNF and DNF with Some Assumed Syntactic Properties}
\label{sec-cnf-dnf}

Following \cite{vgt} we will consider a notion of range-restriction defined in
terms of properties of two prenex formulas that are equivalent to the original
formula, have both the same quantifier prefix but matrices in CNF and DNF,
respectively.
Although not syntactically unique, we refer to them functionally as $\cnf{F}$
and $\dnf{F}$ since we only rely on specific -- easy to achieve -- syntactic
properties that are stated in the following
Props.~\ref{prop-cnf-dnf-base}--\ref{prop-cnf-dnf-properties}.
\begin{prop}
  \label{prop-cnf-dnf-base}
  For all formulas $F$ it holds that
  $\var{\cnf{F}} \subseteq \var{F}$;\linebreak
  $\voc{\cnf{F}} \subseteq \voc{F}$;
  $\var{\dnf{F}} \subseteq \var{F}$;
  $\voc{\dnf{F}} \subseteq \voc{F}$.
\end{prop}
For prenex formulas $F$ with an NNF matrix let $\dual{F}$ be the formula
obtained from $F$ by switching quantifiers $\forall$ and $\exists$,
connectives $\land$ and $\lor$, truth-value constants $\true$ and $\false$,
and literals with their complement.
\begin{prop}%
  \label{prop-cnf-dnf-duality}
For all formulas $F$ it holds that $\cnf{F} = \dual{\dnf{\lnot F}}$; $\dnf{F}
\linebreak = \dual{\cnf{\lnot F}}$; $\cnf{\lnot F} = \dual{\dnf{F}}$; $\dnf{\lnot F} =
\dual{\cnf{F}}$.
\end{prop}

\begin{prop}%
  \label{prop-cnf-dnf-properties}
Let $F_1, F_2,$ $\ldots, F_n$ be NNF formulas. Then
\slabinline{prop-cnf-sub} Each clause in $\cnf{\bigwedge_{i=1}^{n} F_i}$ is in some
$\cnf{F_j}$.
\slabinline{prop-dnf-sub} Each conjunctive clause in $\dnf{\bigvee_{i=1}^{n} F_i}$
is in some $\dnf{F_j}$.
\slabinline{prop-cnf-unit} Formulas $F_j$ that are literals are in each clause in
$\cnf{\bigvee_{i=1}^{n} F_i}$.
\slabinline{prop-dnf-unit} Formulas $F_j$ that are literals are in each conjunctive
clause in $\dnf{\bigwedge_{i=1}^{n} F_i}$.
\slabinline{prop-crr-preserve} If $S$ is a set of variables such that for all $i \in
\{1,\ldots,n\}$ and clauses~$C$ in $\cnf{F_i}$ it holds that $\vall{C} \cap S
\subseteq \vneg{C}$, then for all clauses~$C$ in $\cnf{\bigvee_{i=1}^{n} F_i}$
it holds that $\vall{C} \cap S \subseteq \vneg{C}$.
\slabinline{prop-drr-preserve} If $S$ is a set of variables such that for all $i \in
\{1,\ldots,n\}$ and conjunctive clauses~$D$ in $\dnf{F_i}$ it holds that
$\vall{D} \cap S \subseteq \vpos{D}$, then for all conjunctive clauses~$D$ in
$\dnf{\bigwedge_{i=1}^{n} F_i}$ it holds that $\vall{D} \cap S \subseteq
\vpos{D}$.
\end{prop}

\subsection{Used Notions of Range-Restriction}

The following definition renders the characteristics of the range-restricted
formulas as considered by Van Gelder and Topor in \cite[Theorem~7.2]{vgt}
(except for the special consideration of equality in \cite{vgt}).
\begin{defn}%
  \label{def-rr-vgt}
  A formula $F$ with free variables $\VX$ is called \defname{\vgtrred} if
  $\cnf{F} = Q\, \MC$ and $\dnf{F} = Q\, \MD$, where $Q$ is a quantifier
  prefix (the same in both formulas) upon universally quantified
  variables~$\VU$ and existentially quantified variables~$\VE$ (in arbitrary
  order), and $\MC$, $\MD$ are quantifier-free formulas in CNF and DNF,
  respectively, such that
  \begin{enumerate}
  \item \label{item-def-rr-vgt-u} For all clauses~$C$ in $\MC$ it holds that
    $\vall{C} \cap \VU \subseteq \vneg{C}$.
  \item \label{item-def-rr-vgt-e} For all conjunctive clauses~$D$ in $\MD$ it
    holds that $\vall{D} \cap \VE \subseteq \vpos{D}$.
   \item \label{item-def-rr-vgt-x} For all conjunctive clauses~$D$ in $\MD$ it
     holds that $\VX \subseteq \vpos{D}$.
  \end{enumerate}
\end{defn}
For \vgtrred formulas it is shown in \cite{vgt} that these can be translated
via two intermediate formula classes to a relational algebra expression.
Related earlier results include
\cite{nicolas:1979:early,nicolas:1982,demolombe:1982:early,demolombe:1992}.
The constraint on universal variables is also useful on its own as a weaker
variation of range-restriction, defined as follows.
\begin{defn}%
  \label{def-rr-semi}
  A formula $F$ is called \defname{\unirred} if $\cnf{F} = Q\, \MC$ where $Q$
  is a quantifier prefix upon the universally quantified variables~$\VU$
  (there may also be existentially quantified variables in $Q$) and $\MC$ is a
  quantifier-free formula in CNF such that
  for all clauses~$C$ in $\MC$ it holds
  that $\vall{C} \cap \VU \subseteq \vneg{C}$.
\end{defn}
For formulas without free variables, U-range-restriction and
VGT-range-restric\-tion are related as follows.
\begin{prop}%
  Let $F$ be a sentence. Then
  \slabinline{prop-vgt-ito-semi}
  $F$ is \vgtrred iff $F$ and $\lnot F$ are
  both \unirred.
  \slabinline{ptop-vgt-equi-semi-uni} If $F$ is universal (i.e., in prenex
  form with only universal quantifiers), then $F$ is \vgtrred iff $F$ is
  \unirred.
  \slabinline{ptop-vgt-equi-semi-exi} If $F$ is existential (i.e., in prenex
  form with only existential quantifiers), then $F$ is \vgtrred iff $\lnot F$
  is \unirred.
\end{prop}

U-range-restriction covers well-known restrictions of knowledge bases and
inputs of bottom-up calculi for first-order logic and fragments of it that are
naturally represented by clausal formulas \cite{bumg}. First-order
representations of tuple-generating dependencies (TGDs) are
VGT-range-restricted sentences: conjunctions of sentences of the form $\forall
\XS\YS\, (A(\XS\YS) \imp \exists \ZS\, B(\YS\ZS))$, where $A$ is a possibly
empty conjunction of relational atoms, $B$ is a nonempty conjunction of
relational atoms and the free variables of $A$ and $B$ are exactly those in
the sequences $\XS\YS$ and $\YS\ZS$, respectively. Also certain
generalizations, e.g., to disjunctive TGDs, where $B$ is built up from atoms,
$\land$ and $\lor$, are VGT-range-restricted.

\subsection{Results on Range-Restricted Interpolation}

The following theorem shows three variations for obtaining range-restricted
interpolants from range-restricted inputs.
\begin{thm}[Interpolation and Range-Restriction]
  \label{thm-ipol-rr}
  Let $F$ and $G$ be formulas such that $F \entails G$.

  \smallskip

  \slab{thm-ipol-semi} If $F$ is \unirred, then there exists a
  \univ-range-restricted Craig-Lyndon interpolant~$H$ of $F$ and~$G$.
  Moreover, $H$ can be effectively constructed from
  a clausal tableau proof of $F \entails G$.

  \smallskip

  \slab{thm-ipol-vgt} If $F$ and $G$ are sentences such that $F$ and $\lnot G$
  are \unirred, then there exists a \vgtrred Craig-Lyndon interpolant~$H$
  of $F$ and $G$. Moreover, $H$ can be effectively constructed from
  a clausal tableau proof of $F \entails G$.

  \smallskip

  \slab{thm-ipol-vgt-vx} If $F$ and $\lnot G$ are \unirred, $\var{F} = \var{G}
  = \VX$, and (1)~no clause in $\cnf{F}$ has only negative literals; (2)~for
  all clauses~$C$ in $\cnf{\lnot G}$ with only negative literals it holds that
  $\VX \subseteq \vneg{C}$; (3)~for all clauses~$C$ in $\cnf{\lnot G}$ it
  holds that $\vall{C} \cap \VX \subseteq \vneg{C}$, then there
  exists a \vgtrred
  Craig-Lyndon interpolant~$H$ of $F$ and $G$. Moreover, $H$
  can be effectively constructed from a clausal tableau proof
  of $F \entails G$.
\end{thm}
Observe that Theorem~\ref{thm-ipol-semi} requires range-restriction only
for~$F$, the first of the two interpolation arguments.
Theorem~\ref{thm-ipol-vgt-vx} aims at applications for query reformulation
that in a basic form are expressed as interpolation task for input formulas $F
= K \land Q(\XS)$ and $G = \lnot K' \lor Q'(\XS)$. Here $K$ expresses
background knowledge and constraints as a \unirred sentence and
$Q(\XS)$ represents a query to be reformulated, with free variables $\XS$.
Formulas $K'$ and $Q'$ are copies of $K$ and $Q$, respectively, where
predicates not allowed in the interpolant are replaced by primed versions. If
the query $Q$ is Boolean, i.e., $\XS$ is empty, and $Q$ is \vgtrred, then
Theorem~\ref{thm-ipol-vgt} already suffices to justify the construction of a
\vgtrred interpolant.
If~$\XS$ is not empty, the fine-print preconditions of
Theorem~\ref{thm-ipol-vgt-vx} come into play. Precondition~(1) requires that
$\cnf{K}$ does not have a clause with only negative literals, which is
satisfied if $K$ represents TGDs. Also $\cnf{Q}$ is not allowed to have a
clause with only negative literals. By precondition~(2) all the free variables
$\XS$ must occur in all those clauses of $\cnf{\lnot Q}$ that only have
negative literals, which follows if $Q$ meets
condition~(\ref{item-def-rr-vgt-x}.) of the \vgtrrion (Def.~\ref{def-rr-vgt}).
By precondition~(3) for all clauses~$C$ in $\cnf{\lnot Q}$ it must hold that
$\vall{C} \cap \XS \subseteq \vneg{C}$. A sufficient condition for~$Q$ to meet
all these preconditions is that $\dnf{Q}$ has a purely existential quantifier
prefix and a matrix with only positive literals where each query variable,
i.e., member of $\XS$, occurs in each conjunctive clause.

\subsection{Proving Range-Restricted Interpolation -- The Hyper Property}
\label{sec-proving-rr}

We will prove Theorem~\ref{thm-ipol-rr} by showing how the claimed
interpolants can be obtained with \CTIF. As a preparatory step we match items
from the specification of \CTIF (Fig.~\ref{fig-ctif}) with the constraints of
range-restriction. The following notion gathers intermediate formulas and sets
of symbols of \CTIF.
\begin{defn}%
 An \defname{interpolation context} is a tuple $\ipctx$, where $F,G$ are
 formulas, $\FC, \NGC$ are clausal formulas, $\VC$ is a set of constants,
 $\ffs, \ggs$ are sets of functions, and $\VF, \VG, \VV$ are sets of terms
 such that the following holds.
(i)~$F \entails G$.
(ii)~Let $F_c$ and $G_c$ be $F$ and $G$ after replacing each free variable
with a dedicated fresh constant. Let $\VC$ be those constants that were used
there to replace a variable that occurs in both $F$ and $G$. $\FC$ and $\NGC$
are the matrices of $\cnf{F_c}$ and of $\cnf{\lnot G_c}$, after replacing
existentially quantified variables with Skolem terms.
(iii)~$\ffs$ is the union of the set of the Skolem functions introduced for
existential quantifiers of $\cnf{F_c}$, the set of functions occurring in
$F_c$ but not in $G_c$ and, possibly, further functions freshly introduced in
the grounding step of \CTIF. Analogously, $\ggs$ is the union of the set of
the Skolem functions introduced for $\cnf{\lnot G_c}$, the set of functions
occurring in $G_c$ but not in $F_c$, and, possibly, further functions
introduced in grounding.
(iv)~$\VF$ and $\VG$ are the sets of all terms with outermost function symbol
in $\ffs$ and $\ggs$, respectively.
(v)~$\VV$ is $\VF \cup \VG \cup \VC$.
\end{defn}

The following statements about an interpolation context are easy to infer.
\begin{lem}%
  \label{lem-clause-shape}
  Let $\ipctx$ be an interpolation context. Then
  \slabinline{lem-g:c} No member of $\ggs$ occurs in $\FC$.
  \slabinline{lem-f:c} No member of $\ffs$ occurs in $\NGC$.
  \slabinline{lem-f:cv} If $F$ is \unirred, then for all clauses $C$ in $\FC$ it
  holds that if a variable occurs in $C$ in a position that is not within an
  $\VF$-term it occurs in $C$ in a negative literal, in a position that is not
  within an $\VF$-term.
  \slabinline{lem-g:cv} If $\lnot G$ is \unirred, then for all clauses $C$ in $\NGC$
  it holds that if a variable occurs in $C$ in a position that is not within an
  $\VG$-term, it occurs in $C$ in a negative literal, in a position that is not
  within an $\VG$-term.
  \slabinline{lem-x:cv} If $G$ satisfies condition~(3) of
  Theorem~\ref{thm-ipol-vgt-vx}, then for all clauses $C$ in $\NGC$ it holds
  that any member of $\VC$ that occurs in $C$ in a position that is not within
  an $\VG$-term occurs in $C$ in a negative literal in a position that is not
  within an $\VG$-term.
\end{lem}

\CTIF involves conversion of terms to variables at lifting
(step~\ref{step-lifting}) and at replacing placeholder constants
(step~\ref{step-constants-to-vars}). We introduce a notation to identify those
terms that will be converted there to variables. It mimics the notation for
the set of free variables of a formula but applies to a set of terms, those
with occurrences that are ``maximal'' with respect to a given set~$S$ of
terms, i.e., are not within another term from~$S$.
For NNF formulas $F$ define $S\xtmax{F}$ as the set of $S$-terms that occur in
$F$ in a position other than as subterm of another $S$-term. Define
$S\xtmaxpos{F}$ ($S\xtmaxneg{F}$, respectively) as the set of $S$-terms that
occur in $F$ in a positive (negative, respectively) literal in a position
other than as subterm of another $S$-term. We can now conclude from
Lemma~\ref{lem-clause-shape} the following properties of instances of clauses
used for interpolant construction.
\begin{lem}%
  \label{lem-clause-inst}
  Let $\ipctx$ be an interpolation context. Then

  \slab{lem-i:u} If $F$ is \unirred, then for all instances $C$
  of a clause in $\FC$ it holds that $\vtmax{C} \cap \VG \subseteq
  \vtmaxneg{C}$.

  \slab{lem-d:i:e} If $\lnot G$ is \unirred, then for all
  instances $C$ of a clause in $\NGC$ it holds that $\vtmax{C} \cap \VF
  \subseteq \vtmaxneg{C}$.

  \slab{lem-x:i:noneg} If condition~(1) of Theorem~\ref{thm-ipol-vgt-vx}
  holds, then no instance~$C$ of a clause in $\FC$ has only negative literals.

  \slab{lem-x:i:neg} If condition~(2) of Theorem~\ref{thm-ipol-vgt-vx} holds,
  then for all instances $C$ of a clause in $\NGC$ with only negative literals
  it holds that $\VC \subseteq \vtmaxneg{C}$.

  \slab{lem-x:i:x} If $\lnot G$ is \unirred and condition~(3) of
  Theorem~\ref{thm-ipol-vgt-vx} holds, then for all instances $C$ of a clause
  in $\NGC$ it holds that $\vtmax{C} \cap \VC \subseteq \vtmaxneg{C}$.

\end{lem}

The following proposition adapts Props.~\ref{prop-crr-preserve}
and~\ref{prop-drr-preserve} to $S\tmaxfun$.
\begin{prop}
  \label{prop-cnf-dnf-properties-termadapt}
  Let $F_1, F_2,$ $\ldots, F_n$ be NNF formulas and let $T$ be a set of terms.
  Then
    \slabinline{prop-crr-preserve-terms} If $S$ is a set of terms such that for all
    $i \in \{1,\ldots,n\}$ and clauses~$C$ in $\cnf{F_i}$ it holds that
    $T\xtmax{C} \cap S \subseteq T\xtmaxneg{C}$, then for all clauses~$C$ in
    $\cnf{\bigvee_{i=1}^{n} F_i}$ it holds that $T\xtmax{C} \cap S \subseteq
    T\xtmaxneg{C}$.
    \slabinline{prop-drr-preserve-terms} If $S$ is a set of terms such that for all
    $i \in \{1,\ldots,n\}$ and conjunctive clauses~$D$ in $\dnf{F_i}$ it holds
    that $T\xtmax{D} \cap S \subseteq T\xtmaxpos{D}$, then for all conjunctive
    clauses~$D$ in $\dnf{\bigwedge_{i=1}^{n} F_i}$ it holds that $T\xtmax{D} \cap
    S \subseteq T\xtmaxpos{D}$.
\end{prop}

The key to obtain range-restricted interpolants from \CTIF is that the tableau
must have a specific form, which we call \name{hyper}, as it resembles proofs
by hyperresolution \cite{robinson:hyper} and hypertableaux
\cite{hypertableaux}.
\begin{defn}
  A clausal tableau is called \defname{hyper} if the nodes labeled with a
  negative literal are exactly the leaf nodes.
\end{defn}
While hyperresolution and related approaches, e.g.,
\cite{robinson:hyper,satchmo,puhr,hypertableaux,bumg}, consider DAG-shaped
proofs with non-rigid variables, aiming at interpolant extraction we consider
the hyper property for tree-shaped proofs with rigid variables. The
\name{hyper} requirement is w.l.o.g. because arbitrary closed clausal tableaux
can be converted to tableaux with the hyper property, as we will see in
Sect.~\ref{sec-hyper-convert}.

The proof of Theorem~\ref{thm-ipol-rr} is based on three properties that
invariantly hold for all nodes, or for all inner nodes, respectively, stated
in the following lemma.
\begin{lem}%
  \label{lem-inv-all}
  Let $\ipctx$ be an interpolation context and assume a leaf-closed and hyper
  two-sided clausal ground tableau for $\FC$ and $\NGC$.

  \slab{lem-inv-c} If $F$ is \unirred, then for all nodes $N$ the property
  $\INVC{N}$ defined as follows holds: $\INVC{N}$ $\eqdef$ For all clauses $C$
  in $\cnf{\nipol{N}}$ it holds that
  $\vtmax{C} \cap \VG \subseteq \vtmaxneg{C} \cup \vtmaxpos{\npath{\LL}{N}}.$

  \slab{lem-inv-d} If $\lnot G$ is \unirred, then for all nodes $N$ the
  property $\INVD{N}$ defined as follows holds: $\INVD{N}$ $\eqdef$ For all
  conjunctive clauses $D$ in $\dnf{\nipol{N}}$ it holds that
  $\vtmax{D} \cap \VF \subseteq \vtmaxpos{D} \cup \vtmaxpos{\npath{\RR}{N}}.$

  \slab{lem-inv-x} If $\lnot G$ is \unirred and conditions~(1)--(3)
  Theorem~\ref{thm-ipol-vgt-vx} hold, then for all inner nodes $N$ the
  property $\INVX{N}$ defined as follows holds: $\INVX{N}$ $\eqdef$ For all
  conjunctive clauses $D$ in $\dnf{\nipol{N}}$ it holds that $\VC \subseteq
  \vtmaxpos{D} \cup \vtmaxpos{\npath{\RR}{N}}.$
\end{lem}

Each of Lemma~\ref{lem-inv-c}, \ref{lem-inv-d} and~\ref{lem-inv-x} can be
proven independently by an induction on the tableau structure, but for the
same tableau, such that the properties claimed by them can be combined. In
proving these three sub-lemmas it is sufficient to use their respective
preconditions only to justify the application of matching sub-lemmas of
Lemma~\ref{lem-clause-inst}. That lemma might thus be seen as an abstract
interface that delivers everything that depends on these preconditions and is
relevant for Theorem~\ref{thm-ipol-rr}.

We show here the proof of Lemma~\ref{lem-inv-c}. Lemma~\ref{lem-inv-d} can be
proven in full analogy. The proof of Lemma~\ref{lem-inv-x} is deferred to
\appref{app:proof:lem:inv:x}.
In general, recall that the tableau in Lemma~\ref{lem-inv-all} is a two-sided
tableau for $\FC$ and $\NGC$ that is leaf-closed and hyper. Hence literal
labels of leaves are negative, while those of inner nodes are positive. All
tableau clauses are ground and with an associated \name{side} in $\{\LL,\RR\}$
such that a tableau clause with side $\LL$ is an instance of a clause in $\FC$
and one with side $\RR$ is an instance of a clause in $\NGC$.

\begin{proof}[Lemma~\ref{lem-inv-c}]
  \prlReset{lem-inv-c} By induction on the tableau structure.

\textit{Base case where $N$ is a leaf.} If $N$ and $\ntgt{N}$ have the same
side, then $\nipol{N}$ is a truth value constant, hence $\vtmax{\nipol{N}} =
\emptyset$, implying $\INVC{N}$. If~$N$ has side $\LL$ and $\ntgt{N}$ has side
$\RR$, then $\nipol{N} = \nlit{N}$, which, because $N$ is a leaf, is a
negative literal. Thus $\vtmax{\nipol{N}} = \vtmaxneg{\nipol{N}}$, which
implies $\INVC{N}$. If~$N$ has side $\RR$ and $\ntgt{N}$ has side $\LL$, then
$\nipol{N} = \nlit{\ntgt{N}}$, which, because $N$ is a leaf, is a positive
literal. Thus $\vtmax{\nipol{N}} \subseteq\linebreak
\vtmaxpos{\npath{\LL}{N}}$, implying $\INVC{N}$.

\textit{Induction Step.}  Let $N_1, \ldots, N_n$, where $1 \leq n$, be the
children of $N$.  Assume as induction hypothesis that for $i \in
\{1,\ldots,n\}$ it holds that $\INVC{N_i}$.  Consider the case where the side
of the children is $\LL$.  Then
    \begin{enumerate}[left=\prenum]
    \item[\prl{pr:n:disj}] $\nipol{N} = \bigvee_{i = 1}^{n}  \nipol{N_i}$.
    \end{enumerate}
    Assume that $\INVC{N}$ does not hold. Then there exists a clause~$K$ in
    $\cnf{\nipol{N}}$ and a term $t$ such that \prl{pr:ni1:universal} $t \in
    \VG$; \prl{pr:ni1:open} $t \in \vtmax{K}$; \prl{pr:ni1:litnew} $t \notin
    \vtmaxneg{K}$; \prl{pr:ni1:brnew}~$t \notin \vtmaxpos{\npath{\LL}{N}}$.
    To derive a contradiction, we first show that given
    \pref{pr:ni1:universal}, \pref{pr:ni1:litnew} and \pref{pr:ni1:brnew} it
    holds that
\begin{enumerate}[left=\prenum]
\item[\prl{pr:lem:children}] For all children $N'$ of $N$:
  $t \notin \vtmaxpos{\npath{\LL}{N'}}$.
\end{enumerate}
Statement~\pref{pr:lem:children} can be proven as follows. Assume to the
contrary that there is a child $N'$ of $N$ such that $t \in
\vtmaxpos{\npath{\LL}{N'}}$.  By \pref{pr:ni1:brnew} it follows that $t \in
\vtmax{\nlit{N'}}$ and $\nlit{N'}$ is positive.  By Lemma~\ref{lem-i:u}
and~\pref{pr:ni1:universal} there is another child $N''$ of $N$ such that
$\nlit{N''}$ is negative and $t \in \vtmax{\nlit{N''}}$.  Since the tableau is
closed, it follows from \pref{pr:ni1:brnew} that $\ntgt{N''}$ has side $\RR$,
which implies that $\nipol{N''} = \nlit{N''}$.  Hence $t \in
\vtmax{\nipol{N''}}$.  Since $\nipol{N''}$ is a negative literal and a disjunct
of $\nipol{N}$, it follows from \pref{pr:n:disj} and Prop.~\ref{prop-cnf-unit}
that for all clauses $C$ in $\cnf{\nipol{N}}$ it holds that $t \in \vtmaxneg{C}$,
contradicting assumption~\pref{pr:ni1:litnew}. Hence~\pref{pr:lem:children}
must hold.

From \pref{pr:lem:children}, \pref{pr:ni1:universal} and the induction
hypothesis it follows that for all children~$N'$ of~$N$ and clauses~$C'$ in
$\cnf{\nipol{N'}}$ it holds that $\vtmax{C'} \cap \{t\} \subseteq \vtmaxneg{C'}$.
Hence, by~\pref{pr:n:disj} and Prop.~\ref{prop-crr-preserve-terms} it follows that
for all clauses~$C$ in $\cnf{\nipol{N}}$ it holds that $\vtmax{C} \cap \{t\}
\subseteq \vtmaxneg{C}$.  This, however, contradicts our assumption of the
existence of a clause $K$ in $\cnf{\nipol{N}}$ that satisfies
\pref{pr:ni1:open} and~\pref{pr:ni1:litnew}.  Hence $\INVC{N}$ must hold.

We conclude the proof of the induction step for $\INVC{N}$ by considering the
case where the side of the children of $N$ is $\RR$.  Then
\begin{enumerate}[left=\prenum]
\item[\prl{pr:n:conj}] $\nipol{N} = \bigwedge_{i = 1}^{n}  \nipol{N_i}$.
\item[\prl{pr:i2:aux}] For all children $N'$ of $N$: $\npath{\LL}{N} =
  \npath{\LL}{N'}$.
\end{enumerate}
$\INVC{N}$ follows from the induction hypothesis, \pref{pr:i2:aux},
\pref{pr:n:conj} and Prop.~\ref{prop-cnf-sub}.
\qed
\end{proof}

The invariant properties of tableau nodes shown in
Lemmas~\ref{lem-inv-c}--\ref{lem-inv-x} apply in particular to the tableau
root. We now apply this to prove Theorem~\ref{thm-ipol-rr}.
\begin{proof}[Theorem~\ref{thm-ipol-rr}]
  Interpolants with the stated properties are obtained with \CTIF, assuming
  w.l.o.g. that the CNF computed in step~\ref{step-pre} meets the requirement
  of Sect.~\ref{sec-cnf-dnf}, and that the closed clausal tableau computed in
  step~\ref{step-tab} is leaf-closed and has the hyper property. That \CTIF
  constructs a Craig-Lyndon interpolant has been shown in \cite{cw:ipol}. It
  remains to show the further claimed properties of the interpolant. Let
  $\ipctx$ be the interpolation context for the input formulas $F$ and $G$ and
  let $N_0$ be the root of the tableau computed in step~\ref{step-tab}. Since
  $N_0$ is the root, $\npath{\LL}{N_0} = \npath{\RR}{N_0} = \true$ and thus
  the expressions $\vtmaxpos{\npath{\LL}{N_0}}$ and
  $\vtmaxpos{\npath{\RR}{N_0}}$ in the specifications of $\INVC{N_0}$,
  $\INVD{N_0}$ and $\INVX{N_0}$ all denote the empty set. The claims made in
  the particular sub-theorems can then be shown as follows.

  (\ref{thm-ipol-semi}) By Lemma~\ref{lem-inv-c} it follows that $\INVC{N_0}$.
  Hence, for all clauses~$C$ in $\cnf{\nipol{N_0}}$ it holds that $\vtmax{C}
  \cap \VU \subseteq \vtmaxneg{C}$. It follows that the result of the
  interpolant lifting (step~\ref{step-lifting}) of \CTIF applied to
  $\nipol{N_0}$ is \unirred. Placeholder constant replacement
  (step~\ref{step-constants-to-vars}) does not alter this.

  (\ref{thm-ipol-vgt}) As for Theorem~\ref{thm-ipol-semi} it follows that for
  all clauses~$C$ in $\cnf{\nipol{N_0}}$ it holds that $\vtmax{C} \cap \VU
  \subseteq \vtmaxneg{C}$. By Lemma~\ref{lem-inv-d} it follows that
  $\INVD{N_0}$. Hence, for all conjunctive clauses~$D$ in $\dnf{\nipol{N_0}}$
  it holds that $\vtmax{D} \cap \VE \subseteq \vtmaxpos{D}$. It follows that
  the result of the interpolant lifting of \CTIF applied to $\nipol{N_0}$ is
  \unirred. Since $F$ and $G$ have no free variables, placeholder constant
  replacement has no effect.

  (\ref{thm-ipol-vgt-vx}) As for Theorem~\ref{thm-ipol-vgt} it follows that
  for all clauses~$C$ in $\cnf{\nipol{N_0}}$ it holds that $\vtmax{C} \cap \VU
  \subseteq \vtmaxneg{C}$ and for all conjunctive clauses~$D$ in
  $\dnf{\nipol{N_0}}$ it holds that $\vtmax{D} \cap \VE \subseteq
  \vtmaxpos{D}$. By Lemma~\ref{lem-inv-x} it follows that $\INVX{N_0}$. Hence,
  for all conjunctive clauses~$D$ in $\dnf{\nipol{N_0}}$ it holds that $\VC
  \subseteq \vtmaxpos{D}$. It follows that the result of the interpolant
  lifting of \CTIF applied to $\nipol{N_0}$ followed by placeholder constant
  replacement, now applied to $\VC$, is \vgtrred. \qed
\end{proof}

\section{Horn Interpolation}
\label{sec-horn}

A \defname{Horn clause} is a clause with at most one positive literal. A
\defname{Horn formula} is built up from Horn clauses with the connectives
$\land$, $\exists$ and $\forall$. Horn formulas are important in countless
theoretical and practical respects. Our interpolation method on the basis of
clausal tableaux with the hyper property can be applied to obtain a Horn
interpolant under the precondition that the first argument formula $F$ of the
interpolation problem is Horn. The following theorem makes this precise. It
can be proven by an induction on the structure of a clausal tableau with the
hyper property (see \appref{app:horn}).

\begin{thm}[Interpolation from a Horn Formula]
  \label{thm-ipol-horn}
  Let $F$ be a Horn formula and let $G$ be a formula such that $F\entails G.$
  Then there exists a Craig-Lyndon interpolant~$H$ of $F$ and $G$ that is a
  Horn formula. Moreover, $H$
  can be effectively constructed from a clausal tableau proof
  of $F \entails G$.
\end{thm}

An apparently weaker property than Theorem~\ref{thm-ipol-horn} has been shown
in \cite[\S~4]{mcnulty:uhorn} with techniques from model theory: For
\emph{two} universal Horn formulas $F$ and $G$ there exists a universal Horn
formula that is like a Craig interpolant, except that function symbols are not
constrained. A \name{universal} Horn formula is there a prenex formula with
only universal quantifiers and a Horn matrix. For \CTIF, the corresponding
strengthening of the interpolant to a universal formula can be read-off from
the specification of interpolant lifting (step~\ref{step-lifting} in
Fig.~\ref{fig-ctif}).

The following corollary shows that Theorem~\ref{thm-ipol-horn} can be
combined with Theorem~\ref{thm-ipol-rr} to obtain interpolants that are both
Horn and range-restricted.
\begin{coro}[Range-Restricted Horn Interpolants]
  \label{coro-rr-horn}
  Theorems~\ref{thm-ipol-semi}, \ref{thm-ipol-vgt} and~\ref{thm-ipol-vgt-vx}
  can be strengthened: If $F$ is a Horn formula, then there
  exists a Craig-Lyndon
  interpolant~$H$ with the properties shown in the respective theorem and the
  additional property that it is Horn. Moreover, $H$
  can be effectively constructed from a clausal tableau proof
  of $F \entails G$.
\end{coro}

\begin{proof}
  Can be shown by combining the proof of Theorem~\ref{thm-ipol-semi},
  \ref{thm-ipol-vgt} and~\ref{thm-ipol-vgt-vx}, respectively, with the proof
  of interpolation from a Horn sentence, Theorem~\ref{thm-ipol-horn}. The
  combined proofs are based on inductions on the same closed tableau with
  the hyper property. \qed
\end{proof}

\section{Obtaining Proofs with the Hyper Property}
\label{sec-hyper-convert}

\renewcommand{\tableauscale}{0.7}

Our new interpolation theorems, Theorems~\ref{thm-ipol-rr}
and~\ref{thm-ipol-horn}, depend on the hyper property of the underlying closed
clausal tableaux from which interpolants are extracted. We present a proof
transformation that converts any closed clausal tableau to one with the hyper
property. The transformation can be applied to a clausal tableau as obtained
directly from a clausal tableaux prover. Moreover, it can be also be
indirectly applied to a resolution proof. To this end, the resolution
deduction \emph{tree} \cite{chang:lee} of the binary resolution proof is first
translated to a closed clausal ground tableau in \name{cut normal form}
\cite[Sect. 7.2.2]{letz:habil}.
There the inner
clauses are atomic cuts, tautologies of the form $\lnot p(t_1,\ldots,t_n) \lor
p(t_1,\ldots,t_n)$ or $p(t_1,\ldots,t_n) \lor \lnot p(t_1,\ldots,t_n)$,
corresponding to literals upon which a (tree) resolution step has been
performed. Clauses of nodes whose children are leaves are instances of input
clauses. Our hyper conversion can then be applied to the tableau in cut normal
form. It is easy to see that a regular leaf-closed tableau with the hyper
property can not have atomic cuts. Hence the conversion might be viewed as an
elimination method for these cuts.

We specify the hyper conversion in Fig.~\ref{fig-proc-hyper} as a procedure
that destructively manipulates a tableau. A \defname{fresh copy} of an ordered
tree~$T$ is there an ordered tree~$T^\prime$ with fresh nodes and edges,
related to $T$ through a bijection~$c$ such that any node $N$ of $T$ has the
same labels (literal label and side label) as node $c(N)$ of $T^\prime$ and
such that the $i$-th edge originating in node $N$ of $T$ ends in node~$M$ if
and only if the $i$-th edge originating in node $c(N)$ of $T^\prime$ ends in
node $c(M)$. The procedure is performed as an iteration that in each round
chooses an inner node with negative literal label and then modifies the
tableau. Hence, at termination there is no inner node with negative literal,
which means that the tableau is hyper. Termination of the procedure can be
shown with a measure that strictly decreases in each round
(Prop.~\ref{proc-hyper-terminates} in \appref{app:proc:hyper:terminates}).
Figures~\ref{fig-hyper-conv-1} and~\ref{fig-hyper-conv-2} show example
applications of the procedure.

\begin{figure}
  \centering
  \fbox{
  \begin{minipage}{0.96\textwidth}
\algoinput A closed clausal tableau.

\algoskip

\algomethod Simplify the tableau to leaf-closing and regular form
(Sect.~\ref{sec-tableaux}). Repeat the following operations until the
resulting tableau is hyper.
\begin{enumerate}
\item \label{step-proc-pick} Let $N^\prime$ be the first node visited in
  pre-order with a child that is an inner node with a negative literal
  label. Let $N$ be the leftmost such child.

\item Create a fresh copy~$U$ of the subtree rooted at $N^\prime$. In~$U$
  remove the edges that originate in the node corresponding to $N$.

\item \label{step-proc-attach} Replace the edges originating in $N^\prime$
  with the edges originating in $N$.

\item \label{step-proc-fix} For each leaf descendant~$M$ of $N^\prime$ with
  $\nlit{M} = \du{\nlit{N}}$: Create a fresh copy~$U^\prime$ of $U$. Change
  the origin of the edges originating in the root of $U^\prime$ to~$M$.

\item \label{step-proc-simp} Simplify the tableau to leaf-closing and regular
  form (Sect.~\ref{sec-tableaux}).
\end{enumerate}

\algooutput A leaf-closed, regular and hyper clausal tableau whose clauses are
clauses of the input tableau.
\end{minipage}
  }

\caption{The \name{hyper conversion} proof transformation procedure.}
\label{fig-proc-hyper}
\end{figure}

\begin{figure}
  \centering
\begin{tikzpicture}[scale=0.85, %
    baseline=(a.north), sibling distance=4em,level distance=7ex, every
    node/.style = {transform shape,anchor=mid}]
  \node (a) {\vbar\textbullet}
  child { node {$\lnot \fq$}
    child { node {$\lnot \fp$}
      child { node {$\fp$} } }
    child { node {$\fq$} }
  };
\end{tikzpicture}
\raisebox{-7ex}{$\;\;\rewrite\;\;$}
\begin{tikzpicture}[scale=0.85, %
    baseline=(a.north), sibling distance=4em,level distance=7ex, every
    node/.style = {transform shape,anchor=mid}]
  \node (a) {\vbar\textbullet}
  child { node {$\lnot \fp$}
    child { node {$\fp$} }
  }
  child { node {$\fq$}
    child { node {$\lnot \fq$} }
  };
\end{tikzpicture}
\raisebox{-7ex}{$\;\;\rewrite\;\;$}
\begin{tikzpicture}[scale=0.85, %
    baseline=(a.north), sibling distance=4em,level distance=7ex, every
    node/.style = {transform shape,anchor=mid}]
  \node (a) {\vbar\textbullet}
  child { node {$\fp$}
    child { node {$\lnot \fp$}
    }
    child { node {$\fq$}
      child { node {$\lnot \fq$} }
    }
  };
\end{tikzpicture}
\caption{Hyper conversion of a closed clausal tableau in two rounds.}
\label{fig-hyper-conv-1}
\end{figure}

\begin{figure}
  \centering
\begin{tikzpicture}[scale=0.85, %
    baseline=(a.north), sibling distance=4em,level distance=6ex, every
    node/.style = {transform shape,anchor=mid}]
  \node (a) {\vbar\textbullet}
  child { node {$\lnot \fq$}
    [sibling distance=3em]
    child { node {$\lnot \fp$}
      child { node {$\fp$} } }
    child { node {$\fp$}
      [sibling distance=2em]
      child { node {$\lnot \fp$} }
      child { node {$\fq$} }
    }
  }
  child { node {$\fq$}
    child { node {$\lnot \fq$} }
  };
\end{tikzpicture}
\raisebox{-7ex}{$\;\rewrite\;$}
\begin{tikzpicture}[scale=0.85, %
    baseline=(a.north), sibling distance=4em,level distance=6ex, every
    node/.style = {transform shape,anchor=mid}]
  \node (a) {\vbar\textbullet}
  child { node {$\lnot \fp$}
    child { node {$\fp$} }
  }
  child { node {$\fp$}
    [sibling distance=3em]
    child { node {$\lnot \fp$} }
    child { node {$\fq$}
      child { node {$\lnot \fq$} }
      child { node {$\fq$}
        child { node {$\lnot \fq$} }
    }}
  };
\end{tikzpicture}
\hspace{-1.0em}%
\raisebox{-7ex}{$\rewritereg\;$}
\begin{tikzpicture}[scale=0.85, %
    baseline=(a.north), sibling distance=4em,level distance=6ex, every
    node/.style = {transform shape,anchor=mid}]
  \node (a) {\vbar\textbullet}
  child { node {$\lnot \fp$}
    child { node {$\fp$} }
  }
  child { node {$\fp$}
    [sibling distance=3em]
    child { node {$\lnot \fp$} }
    child { node {$\fq$}
      child { node {$\lnot \fq$} }
    }
  };
\end{tikzpicture}
\hspace{-0.4em}%
\raisebox{-7ex}{$\rewrite$}
\hspace{-0.4em}%
\begin{tikzpicture}[scale=0.85, %
    baseline=(a.north), sibling distance=3em,level distance=6ex, every
    node/.style = {transform shape,anchor=mid}]
  \node (a) {\vbar\textbullet}
  child { node {$\fp$}
    child { node {$\lnot \fp$}
    }
    child { node {$\fp$}
      child { node {$\lnot \fp$} }
      child { node {$\fq$}
        child { node {$\lnot \fq$} }
      }
    }
  };
\end{tikzpicture}
\hspace{-0.4em}%
\raisebox{-7ex}{$\rewritereg$}
\hspace{-0.4em}%
\begin{tikzpicture}[scale=0.85, %
    baseline=(a.north), sibling distance=3em,level distance=6ex, every
    node/.style = {transform shape,anchor=mid}]
  \node (a) {\vbar\textbullet}
  child { node {$\fp$}
    child { node {$\lnot \fp$}
    }
    child { node {$\fq$}
      child { node {$\lnot \fq$} }
    }
  };
\end{tikzpicture}
\caption{Hyper conversion of a closed clausal tableau in cut normal form in
  two rounds. For each round the result after procedure
  steps~\ref{step-proc-pick}--\ref{step-proc-fix} is shown and then the result
  after step~\ref{step-proc-simp}, simplification, applied here to achieve
  regularity.}
\label{fig-hyper-conv-2}
\end{figure}

Since the hyper conversion procedure copies parts of subtrees it is not a
polynomial operation.\footnote{A thorough complexity analysis should take
calculus- or strategy-dependent properties of the input proofs into account.
And possibly also the blow-up from resolution to tree resolution underlying
the cut normal form tableaux.} To get an idea of its practical feasibility, we
experimented with an unbiased set of proofs of miscellaneous problems. For
this we took those 112 \name{CASC-J11} \cite{casc:2022} problems that could be
proven with \ProverN \cite{prover9} in 400~s per problem, including a basic
proof conversion with \ProverN's tool \Prooftrans.\footnote{On a Linux
notebook with 12th Gen Intel\textsuperscript{\textregistered}
Core\texttrademark\ i7-1260P CPU and 32~GB~RAM.} The hyper conversion
succeeded on 107 (or 96\%) of these, given 400~s timeout per proof, where the
actual median of used time was only 0.01~s. It was applied to a tableau in cut
normal form that represents the proof tree of \ProverN's proof. The two
intermediate steps, translation of paramodulation to binary resolution and
expansion to cut normal form, succeeded in fractions of a second, except for
one case where the expansion took 121~s and two cases where it failed due to
memory exhaustion. The hyper conversion then failed in three further cases.
For all except two proofs the hyper conversion reduced the proof size, where
the overall median of the size ratio hyper-to-input was 0.39. See
\appref{app:exp} for details.

\section{Conclusion}
\label{sec-conclusion}

We conclude with discussing related work, open issues and perspectives. Our
interpolation method \CTIF \cite{cw:ipol} is complete for first-order logic
with function symbols. \Vampire's native interpolation
\cite{vampire:interpol:2010,vampire:interpol:2012}, targeted at verification,
is like all local methods incomplete \cite{kovacs:17}. \Princess
\cite{princess08,brillout:etal:beyond:2011} implements interpolation with a
sequent calculus that supports theories for verification and permits
uninterpreted predicates and functions. Suitable proofs for our approach can
currently be obtained from \CMProver (clausal tableaux) and \ProverN
(resolution/paramodulation). With optimized settings, \Vampire \cite{vampire}
and \EProver \cite{eprover} as of today only output proofs with gaps. This
seems to improve \cite{schulz:prague:2023} or might be overcome by re-proving
with \ProverN using lemmas from the more powerful systems.

So far we did not address special handling of equality in the context of
range-restriction, a topic on its own, e.g., \cite{vgt,bumg}. We treat it as
predicate, with axioms for reflexivity, symmetry, transitivity and
substitutivity. \CTIF works smoothly with these, respecting polarity
constraints of equality in interpolants \cite[Sect.~10.4]{cw:ipol}. With
exception of reflexivity these axioms are U-range-restricted. We do not
interfere with the provers' equality handling and just translate in finished
proofs paramodulation into binary resolution with substitutivity axioms.

The potential bottleneck of conversion to clausal form in \CTIF may be
remedied with structure-preserving (aka \name{definitional}) normal forms
\cite{scott:twovars,tseitin,eder:def:85,plaisted:greenbaum}.

Our \name{hyper} property might be of interest for proof presentation and
exchange, since it gives the proof tree a constrained shape and in experiments
often shortens it. Like hyperresolution and hypertableaux it can be
generalized to take a ``semantics'' into account \cite{slagle:1967}
\cite[Chap.~6]{chang:lee} \cite[Sect.~4.5]{handbook:ar:haehnle}. To shorten
interpolants, it might be combined with proof reductions (e.g.,
\cite{cwwb:lukas:2021}).

For query reformulation, interpolation on the basis of general first-order ATP
was so far hardly considered. Most methods are sequent calculi
\cite{toman:wedell:book,bwp:pods:2023} or analytic tableaux systems
\cite{franconi:2013,toman:2015:tableaux,benedikt:book,toman:2017}.
Experiments with ATP systems and propositional inputs indicate that
requirements are quite different from those in verification
\cite{benedikt:2017}.
An implemented system \cite{toman:2015:tableaux,toman:2017} uses analytic
tableaux with dedicated refinements for enumerating alternate
proofs/interpo\-lants corresponding to query plans for heuristic choice. In
\cite{benedikt:book} the focus is on interpolants that are sentences
respecting binding patterns, which, like range-restriction, ensures database
evaluability.
Our interpolation theorems show fine-grained conditions for passing variations
of range-restriction and the Horn property on to interpolants.
Matching these with the many formula classes considered in knowledge
representation and databases is an issue for future work. A further open topic
is adapting recent synthesis techniques for nested relations
\cite{bwp:pods:2023} to the clausal tableaux proof system.

Methodically, we exemplified a way to approach operations on proof structures
while taking efficient automated first-order provers into account. Feasible
implementations are brought within reach, for practical application and also
for validating abstract claims and conjectures with scrutiny. The prover is a
black box, given freedom on optimizations, strategy and even calculus. For
interfacing, the overall setting incorporates clausification and
Skolemization. Requirements on the proof structure do not hamper proof search,
but are ensured by transformations applied to proofs returned by the efficient
systems.

\subsubsection{Acknowledgments.}
 The author thanks Michael Benedikt for bringing the subtleties of
 range-restriction in databases to attention, Cécilia Pradic for insights into
 subtleties of proof theory, and anonymous reviewers for helpful suggestions
 to improve the presentation.

\bibliographystyle{splncs04}
\bibliography{bibrange01}

\clearpage
\appendix

\section{Proof of Lemma~\ref{lem-inv-x}}
\label{app:proof:lem:inv:x}

This appendix supplements Sect.~\ref{sec-proving-rr} with the proof of
Lemma~\ref{lem-inv-x}, used for proving Theorem~\ref{thm-ipol-vgt-vx} on
range-restricted interpolation with free variables.

\begin{lemref}{}{\ref{lem-inv-all}}
  Let $\ipctx$ be an interpolation context and assume a leaf-closed and hyper
  two-sided clausal ground tableau for $\FC$ and $\NGC$.

  \setcounter{refsub}{2}
  
  \slabref{lem-inv-x} If $\lnot G$ is \unirred and conditions~(1)--(3)
  Theorem~\ref{thm-ipol-vgt-vx} hold, then for all inner nodes $N$ the
  property $\INVX{N}$ defined as follows holds: $\INVX{N}$ $\eqdef$ For all
  conjunctive clauses $D$ in $\dnf{\nipol{N}}$ it holds that $\VC \subseteq
  \vtmaxpos{D} \cup \vtmaxpos{\npath{\RR}{N}}.$
\end{lemref}

\smallskip

The proof of Lemma~\ref{lem-inv-x} proceeds by induction on the tableau
structure, referring to DNF conversions of the interpolant constituents,
similarly to the proof of Lemma~\ref{lem-inv-d}, but with a base case that
resides on the following lemma:
\begin{lem}%
  \label{lem-above-negative-clause}
  For all inner nodes $N$ of a closed tableau that is hyper it holds that
  either all literals in $\nclause{N}$ are negative or $N$ has a descendant
  $N'$ such that all literals in $\nclause{N'}$ are negative.
\end{lem}
\begin{proof} The tableau is hyper, hence leaves are exactly the nodes with
  a negative literal label. Failure of the claimed property would imply that
  the tableau has an infinite branch. \qed
\end{proof}

\begin{proof}[Lemma~\ref{lem-inv-x}]
  \prlReset{lem-inv-x} By induction on the tableau structure, with nodes~$N$
  where all literals in $\nclause{N}$ are negative as base case. That this is
  sufficient as base case to show $\INVX{N}$ for all \emph{inner} nodes $N$ as
  claimed follows from Lemma~\ref{lem-above-negative-clause}.

  \textit{Base case where all literals in $\nclause{N}$ are negative.}  By
  Lemma~\ref{lem-x:i:noneg} the children of $N$ have side $\RR$. Hence
\begin{enumerate}[left=\prenum]
\item[\prl{pr:x:base}] $\nipol{N} = \bigwedge_{i=1}^{n} \nipol{N_i}.$
\end{enumerate}
For all children $N'$ of $N$ where the side of $\ntgt{N'}$ is $\LL$ it holds
that $\nipol{N'} = \nlit{\ntgt{N'}}$, a positive literal, hence
$\vtmax{\nlit{N'}} = \vtmaxpos{\nipol{N'}}$.  For all children $N'$ of $N$ where
the side of $\ntgt{N'}$ is $\RR$ it holds that $\vtmax{\nlit{N'}} \subseteq
\vtmaxpos{\npath{\RR}{N}}$.  With \pref{pr:x:base} it follows that
\begin{enumerate}[left=\prenum]
\item[\prl{pr:x:base:a}] $\vtmax{\nclause{N}} \subseteq \vtmaxpos{\nipol{N}} \cup
  \vtmaxpos{\npath{\RR}{N}}$.
\end{enumerate}
By Lemma~\ref{lem-x:i:neg} it holds that $\VC \subseteq \vtmax{\nclause{N}}$.
With \pref{pr:x:base:a} it follows that
\begin{enumerate}[left=\prenum]
\item[\prl{pr:x:base:b}] $\VC \subseteq \vtmaxpos{\nipol{N}} \cup
  \vtmaxpos{\npath{\RR}{N}}$.
\end{enumerate}
Because $\nipol{N}$ is a conjunction of literals and $\true$, the formula
$\dnf{\nipol{N}}$ consists of a single conjunctive clause~$D$, which contains
exactly the literals in $\nipol{N}$ and thus satisfies $\vtmaxpos{D} =
\vtmaxpos{\nipol{N}}$. Together with \pref{pr:x:base:b} this implies
$\INVX{N}$.

\textit{Induction Step.}  Let $N_1, \ldots, N_n$, where $1 \leq n$, be the
children of $N$.  Assume as induction hypothesis that for $i \in
\{1,\ldots,n\}$ it holds that $\INVX{N_i}$.  Consider the case where the side
of the children is $\RR$.  Then
    \begin{enumerate}[left=\prenum]
    \item[\prl{pr:x:n:disj}] $\nipol{N} = \bigwedge_{i = 1}^{n} \nipol{N_i}$.
    \end{enumerate}
    Assume that $\INVX{N}$ does not hold.  Then there exists a conjunctive
    clause $K$ in $\dnf{\nipol{N}}$ and a term $t$ such that
    \begin{enumerate}[left=\prenum]
    \item[\prl{pr:x:ni1:universal}] $t \in \VC$.
    \item[\prl{pr:x:ni1:litnew}] $t \notin \vtmaxpos{K}$.
    \item[\prl{pr:x:ni1:brnew}]  $t \notin \vtmaxpos{\npath{\RR}{N}}$.
\end{enumerate}
    To derive a contradiction, we first show that given
    \pref{pr:x:ni1:universal}--\pref{pr:x:ni1:brnew} it holds that
\begin{enumerate}[left=\prenum]
\item[\prl{pr:x:lem:children}] For all children $N'$ of $N$:
  $t \notin \vtmaxpos{\npath{\RR}{N'}}$.
\end{enumerate}
Statement~\pref{pr:x:lem:children} can be shown as follows. Assume to the
contrary that there is a child $N'$ of $N$ such that $t \in
\vtmaxpos{\npath{\RR}{N'}}$. By \pref{pr:x:ni1:brnew} it follows that $t \in
\vtmax{\nlit{N'}}$ and $\nlit{N'}$ is positive. By Lemma~\ref{lem-x:i:x}
and~\pref{pr:x:ni1:universal} there is another child $N''$ of $N$ such that
$\nlit{N''}$ is negative and $t \in \vtmax{\nlit{N''}}$. Since the tableau is
closed, it follows from \pref{pr:x:ni1:brnew} that $\ntgt{N''}$ has side
$\LL$, which implies that $\nipol{N''} = \nlit{\ntgt{N''}}$. Hence $t \in
\vtmax{\nipol{N''}}$. Since $\nipol{N''}$ is a positive literal and a conjunct
of $\nipol{N}$, it follows from \pref{pr:x:n:conj} and
Prop.~\ref{prop-dnf-unit} that for all conjunctive clauses $D$ in
$\dnf{\nipol{N}}$ it holds that $t \in \vtmaxpos{D}$, contradicting
assumption~\pref{pr:x:ni1:litnew}. Hence~\pref{pr:x:lem:children} must hold.

From \pref{pr:x:lem:children}, \pref{pr:x:ni1:universal} and the induction
hypothesis it follows that for all children~$N'$ of~$N$ and conjunctive
clauses~$D'$ in $\dnf{\nipol{N'}}$ it holds that $t \in \vtmaxpos{D'}$. Hence,
by~\pref{pr:x:n:disj} and Prop.~\ref{prop-drr-preserve-terms} it follows that
for all clauses~$D$ in $\dnf{\nipol{N}}$ it holds that $t \in \vtmaxpos{D}$.
This, however, contradicts our assumption of the existence of a conjunctive
clause $K$ in $\dnf{\nipol{N}}$ that satisfies~\pref{pr:x:ni1:litnew}. Hence
$\INVX{N}$ must hold.

We conclude the proof of the induction step for $\INVX{N}$ by considering the
case where the side of the children of $N$ is $\LL$. Then
\begin{enumerate}[left=\prenum]
\item[\prl{pr:x:n:conj}] $\nipol{N} = \bigvee_{i = 1}^{n}  \nipol{N_i}$.
\item[\prl{pr:x:i2:aux}] For all children $N'$ of $N$: $\npath{\RR}{N} =
  \npath{\RR}{N'}$.
\end{enumerate}
$\INVX{N}$ follows from the induction hypothesis, \pref{pr:x:i2:aux},
\pref{pr:x:n:conj} and Prop.~\ref{prop-dnf-sub}.
\qed
\end{proof}

\section{Proof of Theorem~\ref{thm-ipol-horn}}
\label{app:horn}

This appendix supplements Sect.~\ref{sec-horn} with a proof of the
central claim there, Theorem~\ref{thm-ipol-horn}.

\begin{thmref}{Interpolation from a Horn Formula}{\ref{thm-ipol-horn}}
  Let $F$ be a Horn formula and let $G$ be a formula such that $F\entails G.$
  Then there exists a Craig-Lyndon interpolant~$H$ of $F$ and $G$ that is a
  Horn formula. Moreover, $H$
  can be effectively constructed from a clausal tableau proof
  of $F \entails G$.
\end{thmref}

\begin{proof}
  We use some additional terminology: A \defname{negative clause} is a clause
  with only negative literals. A \defname{Horn-like formula} is an NNF
  inductively defined as a literal, $\true$, $\false$, a conjunction of
  Horn-like formulas, or a disjunction of negative literals, $\false$ and at
  most one Horn-like formula. It is easy to see that a Horn-like formula can
  be converted to an equivalent conjunction of Horn clauses by truth-value
  simplification and distributing disjunction upon conjunction.

  The claimed Horn interpolant is obtained via \CTIF (Fig.~\ref{fig-ctif}),
  assuming w.l.o.g. that the CNF computed in step~\ref{step-pre} there meets
  the requirement of Sect.~\ref{sec-cnf-dnf}, and that the closed clausal
  tableau computed in step~\ref{step-tab} is leaf-closed and has the hyper
  property. That \CTIF constructs a Craig-Lyndon interpolant has been shown in
  \cite{cw:ipol}. It remains to show that it can be converted to a Horn
  formula. Let $\ipctxprefix$ be the prefix of an interpolation context for
  the input formulas~$F$ and $G$ and let $N_0$ be the root of the tableau
  computed in step~\ref{step-tab}.

  We now show by induction on the tableau structure that $\HG = \nipol{N_0}$
  for the tableau root~$N_0$ is a Horn-like formula. The theorem then follows
  since we can obtain the final interpolant $H$ from $\HG$ by interpolant
  lifting (step~\ref{step-lifting} of \CTIF), replacing placeholder constants
  with free variables (step~\ref{step-constants-to-vars}), and conversion of
  the Horn-like matrix to an equivalent conjunction of Horn clauses, where all
  syntactic properties relevant for a Craig-Lyndon interpolant are preserved.

  For the base case where $N$ is a leaf it is immediate from the definition of
  $\f{ipol}$ that $\nipol{N}$ is a ground literal or a truth value constant
  and thus a Horn-like formula. To show the induction step, let $N$ be an
  inner node with children $N_1,\ldots,N_n$ where $n \geq 1$. As induction
  hypothesis assume that for all $i \in \{1,\ldots,n\}$ it holds that
  $\nipol{N_i}$ is a Horn-like formula. We prove the induction step by showing
  that then also $\nipol{N}$ is a Horn-like formula.
  \begin{itemize}
  \item Case $\nside{N_1} = \aaa$. We consider two subcases.
      \begin{itemize}
      \item Case $\nclause{N}$ is negative. For all $i \in \{1,\ldots,n\}$ the
        formula $\nipol{N_i}$ is a negative literal or $\false$. Hence
        $\nipol{N} = \bigvee_{i=1}^n \nipol{N_i}$ is Horn-like.
      \item Case $\nclause{N}$ is not negative. Since $F'$ is Horn and
        $\nclause{N}$ is an instance of a clause in $F'$, $\nclause{N}$ has
        exactly one child whose literal label is positive. Let $N_j$ with $j
        \in \{1,\ldots,n\}$ be that child. By the induction hypothesis
        $\nipol{N_j}$ is Horn-like. For all $i \in \{1,\ldots,n\} \setminus
        \{j\}$ the formula $\nipol{N_i}$ is a negative literal or $\false$.
        Hence $\nipol{N} = \bigvee_{i=1}^n \nipol{N_i}$ is Horn-like.
      \end{itemize}
    \item Case $\nside{N_1} = \bbb$. From the induction hypothesis it follows
      that $\nipol{N} = \bigwedge_{i=1}^n \nipol{N_i}$ is Horn-like. \qed
  \end{itemize}
\end{proof}

\section{Termination of the Hyper Conversion}
\label{app:proc:hyper:terminates}

This appendix supplements Sect.~\ref{sec-hyper-convert} with a proven
statement on the termination of the hyper conversion procedure
(Fig.~\ref{fig-proc-hyper}).

\begin{prop}%
\label{proc-hyper-terminates}
The hyper conversion procedure (Fig.~\ref{fig-proc-hyper}) terminates.
\end{prop}

\begin{proof}
We give a measure that strictly decreases in each round of the procedure.
Consider a single round of the steps~1--5 of the procedure with $N$ and
$N^\prime$ as determined in step~1. We observe the following.
\begin{enumerate}[label={(\roman*)},leftmargin=2em]
\item \label{item-proc-hyper-allbelow} All tableau modifications made in
  the round are in the subtree rooted at~$N^\prime$.
\item \label{item-proc-hyper-toleaf} At finishing the round all descendants
  of $N^\prime$ with the same literal label as~$N$ are leaves.
\item \label{item-proc-hyper-nonewnonleaf} All literal labels of inner
  nodes that are descendants of $N^\prime$ and are different from
  $\du{\nlit{N}}$ at finishing the round were already literal labels of inner
  nodes that are descendants of~$N^\prime$ when entering the round.
\end{enumerate}
We can now specify the measure that strictly decreases in each round of the
procedure. For a node $N$ define $\nbadlits{N}$ as the set of literal labels
that occur in inner (i.e., non-leaf) descendants of~$N$ and are negative. From
the above observations~\ref{item-proc-hyper-toleaf}
and~\ref{item-proc-hyper-nonewnonleaf} it follows that for $N^\prime$ as
determined in step~1 of the procedure the cardinality of $\nbadlits{N^\prime}$
is strictly decreased in a round of steps~1--5 of the procedure. However, a
different node might be determined as~$N^\prime$ in step~1 of the next round.
To specify a globally decreasing measure we define a further auxiliary notion:
Let $N_n$ be a node whose ancestors are in root-to-leaf order the nodes
$N_1,\ldots,N_{n-1}$. Define $\ncode{N_n}$ as the string $I_1\ldots I_n
\omega$ of numbers, where for $i \in \{1,\ldots,n\}$ the number $I_i$ is the
number of right siblings of $N_i$. With
observation~\ref{item-proc-hyper-allbelow} it then follows that the following
string of numbers, determined in step~1 of a round, is strictly reduced from
round to round w.r.t. the lexicographical order of strings of numbers:
 \[\ncode{N^\prime}|\nbadlits{N^\prime}|.\]
Regularity ensures that the length of the strings to be considered can not be
larger than the finite number of literal labels of nodes of the input tableau
plus~$3$ (a leading $0$ for the root, which has no literal label; $\omega$;
and $|\nbadlits{N^\prime}|$).  With the lexicographical order restricted to
strings up to that length we have a well-order and the strict reduction
ensures termination.
\qed
\end{proof}

\section{Experimental Indicators of Practical Feasibility}
\label{app:exp}

This appendix supplements Sect.~\ref{sec-hyper-convert} with details on the
experiments to verify practical feasibility of proof conversions involved in
our strengthened variations of Craig interpolation. Also instructions for
reproducing the experiments are given.
An implementation of the techniques from the paper is currently in progress,
written in SWI-Prolog \cite{swiprolog}, embedded in \PIE
\cite{cw:pie:2016,cw:pie:2020}. Core parts of the functionality are already
available\footnote{Interpolation with the \CTIF method (\PIE module
  \name{craigtask\_cm}); resolution proof translation and hyper conversion
  (module \name{ctrp}); conversion of VGT-range-restricted formulas to
  \name{allowed} \cite{vgt} formulas (module \name{rr\_support})} but not yet
integrated into full application workflows.

The involved proof transformations lead from a proof with resolution and
paramodulation via pure binary resolution and a clausal tableau in cut normal
form to a clausal tableau with the hyper property. To get an impression of
their practical feasibility, we tested them on problems from the latest
\name{CASC} competition, \name{CASC-J11} \cite{casc:2022}, as an unbiased set
of proofs of miscellaneous problems.

As basis we took those FOF problems of \name{CASC-J11} on which \ProverN
succeeded in the competition. We tried to reprove these in \ProverN's default
\name{auto} mode\footnote{\ProverN's \name{auto2} mode used in the
\name{CASC-J11} competition leads to some proofs with the \name{new\_symbol}
rule that so far can not be translated.} via \PIE and to convert their proofs
with the \Prooftrans tool, which comes with \ProverN, with a timeout of 400~s
per problem.\footnote{On a Linux notebook with 12th Gen
Intel\textsuperscript{\textregistered} Core\texttrademark\ i7-1260P CPU and
32~GB~RAM.} \Prooftrans was configured with the \name{expand} option that
translates \ProverN's proofs to just binary resolution, paramodulation and a
few other equality rules. This succeeded for 112 problems. For one additional
problem, \Prooftrans failed.\footnote{LCL664+1.005, \textit{Fatal error:
set\_vars\_recurse: max\_vars}.} The length of the obtained proofs (number of
steps, including axioms) was between 12 and 919 with median~55.

Equality-specific rules were then translated to binary resolution steps, which
for no proof took longer than 0.04~s. The proof length (number of steps,
including axioms) of the results was between 10 and 4,833, median 81 (in some
cases the size decreased because non-clausal axioms were deleted).

These proofs were then converted to clausal tableaux in cut normal form that
correspond to resolution trees. This failed for two of the 112 proofs due to
memory exhaustion but succeeded for each of the others in less then 0.1~s,
with exception of one problem, where it took~121~s. The median time per
problem was 0.001~s. The proof size (number of inner nodes of the clausal
tableau) in the results was between 20 and 97,866,317, median~259.

Finally the hyper conversion was applied with a timeout of 400~s per proof to
the remaining 110 proofs. It succeeded for 107 proofs, within a median time of
0.01~s per proof, and a maximum time of 235~s. One failure was due to memory
exhaustion, the other two were timeouts. The proof size (number of inner nodes
of the clausal tableau) of the results was between 11 and 3,110, median~77. In
105 of the 107 cases the size was reduced. The largest proof on which the
conversion succeeded had size 51,359 and was reduced to size 507. The ratios
of the size of the hyper-converted tableau to the size of the source tableau
were between 0.01 and 4.48, median 0.39.

Tables~\ref{tab-conv-1}--\ref{tab-conv-3} show result data for each of the
112~problems. Figure~\ref{fig-exp-code} shows Prolog code to reproduce the
experiments with \PIE. The columns in these tables are as follows.

\begin{description}
\item{\textbf{Problem}} The TPTP problem (TPTP v8.1.2).
\item{\textbf{Rtg}} Its latest rating in TPTP v8.1.2.
  \item{\textbf{T1}} Proving time in seconds, rounded (timeout 400~s).
  \item{\textbf{S1}} Number of steps of the original proof, after expansion
    into binary resolution and paramodulation by \Prooftrans, including
    axioms.
  \item{\textbf{S2}} Number of steps of the binary resolution proof after
    translation of paramodulation and some other equality inferences to pure
    binary resolution, including axioms.
  \item{\textbf{S3}} Tree size (number of inner nodes) of the clausal tableau
    in cut normal form, where ``--'' indicates failure of the conversion.
  \item{\textbf{S4}} Tree size (number of inner nodes) of the clausal tableau
    after the hyper conversion, where ``--'' indicates a timeout (400~s) of
    the conversion.
  \item{\textbf{T2}} Time for the hyper conversion in seconds, rounded.
\end{description}

\begin{table}
\centering\small
  \caption{Conversion data for CASC-J11 problems (I/III).}
  \label{tab-conv-1}
\begin{tabular}{L{7em}R{2em}R{2em}R{4em}R{4em}R{5em}R{4em}R{2em}}
  \textbf{Problem} & \textbf{Rtg} & \textbf{T1} & \textbf{S1} &
  \textbf{S2} & \textbf{S3} & \textbf{S4} & \textbf{T2}\\\midrule
ALG019+1 & 0.25 & 0 & 28 & 52 & 129 & 36 & 0\\
ALG210+2 & 0.25 & 0 & 43 & 97 & 305 & 149 & 0\\
ALG219+1 & 0.28 & 1 & 115 & 209 & 305 & 120 & 0\\
COM008+2 & 0.53 & 5 & 60 & 84 & 857 & 104 & 0\\
CSR007+1 & 0.53 & 0 & 217 & 327 & 5,078 & 1,743 & 9\\
CSR008+1 & 0.50 & 0 & 247 & 377 & 8,912 & 2,023 & 21\\
CSR024+1.009 & 0.47 & 1 & 166 & 250 & 865 & 270 & 0\\
CSR024+1.010 & 0.47 & 0 & 179 & 259 & 801 & 265 & 0\\
GEO082+1 & 0.22 & 0 & 22 & 37 & 47 & 22 & 0\\
GEO083+1 & 0.42 & 6 & 24 & 43 & 57 & 27 & 0\\
GEO094+1 & 0.36 & 0 & 40 & 34 & 101 & 40 & 0\\
GEO112+1 & 0.31 & 0 & 21 & 54 & 49 & 21 & 0\\
GEO125+1 & 0.25 & 0 & 29 & 68 & 65 & 27 & 0\\
GEO126+1 & 0.42 & 0 & 66 & 132 & 388 & 372 & 0\\
GEO127+1 & 0.47 & 0 & 92 & 151 & 2,069 & 341 & 0\\
GEO499+1 & 0.39 & 1 & 12 & 28 & 33 & 17 & 0\\
GEO524+1 & 0.39 & 0 & 13 & 27 & 45 & 21 & 0\\
GEO526+1 & 0.56 & 8 & 30 & 60 & 73 & 29 & 0\\
GEO541+1 & 0.33 & 2 & 92 & 73 & 257 & 93 & 0\\
GEO555+1 & 0.39 & 3 & 79 & 65 & 601 & 77 & 0\\
GEO557+1 & 0.39 & 2 & 66 & 52 & 1,203 & 111 & 0\\
GRA007+2 & 0.50 & 399 & 101 & 145 & 1,545 & 331 & 0\\
GRP655+2 & 0.54 & 79 & 919 & 4,833 & -- & -- & --\\
GRP711+1 & 0.21 & 0 & 53 & 116 & 1,001 & 269 & 0\\
GRP720+1 & 0.71 & 7 & 158 & 469 & 10,041 & 2,271 & 183\\
GRP746+1 & 0.58 & 70 & 795 & 2,443 & 97,866,317 & -- & --\\
\end{tabular}
\end{table}

\begin{table}
  \centering\small
  \caption{Conversion data for CASC-J11 problems (II/III).}
  \label{tab-conv-2}
\begin{tabular}{L{7em}R{2em}R{2em}R{4em}R{4em}R{5em}R{4em}R{2em}}
  \textbf{Problem} & \textbf{Rtg} & \textbf{T1} & \textbf{S1} &
  \textbf{S2} & \textbf{S3} & \textbf{S4} & \textbf{T2}\\\midrule
GRP747+1 & 0.25 & 12 & 28 & 59 & 636 & 2,850 & 8\\
GRP779+1 & 0.72 & 6 & 267 & 655 & 11,101 & 1,742 & 105\\
ITP015+1 & 0.44 & 0 & 70 & 250 & 805 & 339 & 1\\
ITP020+1 & 0.42 & 0 & 20 & 67 & 97 & 34 & 0\\
ITP023+1 & 0.36 & 0 & 38 & 105 & 261 & 94 & 0\\
KLE041+1 & 0.44 & 2 & 77 & 134 & 669 & 240 & 0\\
KLE170+1.002 & 0.39 & 0 & 62 & 125 & 2,409 & 220 & 0\\
KRS188+1 & 0.33 & 16 & 16 & 13 & 33 & 17 & 0\\
KRS191+1 & 0.33 & 0 & 13 & 10 & 21 & 11 & 0\\
KRS193+1 & 0.33 & 0 & 13 & 10 & 21 & 11 & 0\\
KRS194+1 & 0.33 & 0 & 13 & 10 & 21 & 11 & 0\\
KRS196+1 & 0.33 & 0 & 13 & 10 & 21 & 11 & 0\\
KRS202+1 & 0.27 & 0 & 45 & 38 & 94 & 48 & 0\\
KRS203+1 & 0.27 & 0 & 35 & 29 & 66 & 34 & 0\\
KRS216+1 & 0.40 & 17 & 21 & 15 & 30 & 17 & 0\\
KRS217+1 & 0.27 & 1 & 32 & 25 & 54 & 27 & 0\\
KRS234+1 & 0.33 & 0 & 19 & 15 & 37 & 16 & 0\\
KRS235+1 & 0.33 & 1 & 19 & 15 & 37 & 16 & 0\\
LCL456+1 & 0.25 & 0 & 25 & 20 & 45 & 22 & 0\\
LCL485+1 & 0.58 & 20 & 96 & 145 & 749 & 307 & 0\\
LCL509+1 & 0.56 & 11 & 99 & 154 & 8,609 & 1,926 & 61\\
LCL549+1 & 0.61 & 150 & 183 & 242 & 1,195 & 1,120 & 5\\
MGT034+2 & 0.22 & 0 & 101 & 143 & 1,306 & 116 & 0\\
MGT039+2 & 0.25 & 0 & 422 & 502 & -- & -- & --\\
MGT042+1 & 0.22 & 0 & 89 & 137 & 667 & 268 & 0\\
MGT061+1 & 0.22 & 0 & 95 & 141 & 529 & 142 & 0\\
MGT065+1 & 0.28 & 0 & 111 & 165 & 3,071 & 556 & 1\\
MGT067+1 & 0.27 & 0 & 58 & 51 & 177 & 64 & 0\\
NUM317+1 & 0.64 & 92 & 27 & 24 & 49 & 25 & 0\\
NUM323+1 & 0.53 & 90 & 18 & 15 & 29 & 15 & 0\\
NUM329+1 & 0.33 & 1 & 15 & 12 & 20 & 12 & 0\\
NUM568+3 & 0.22 & 3 & 15 & 28 & 41 & 17 & 0\\
NUM613+3 & 0.25 & 6 & 58 & 86 & 189 & 81 & 0\\
NUM925+5 & 0.25 & 179 & 90 & 205 & 1,081 & 213 & 0\\
NUM926+2 & 0.47 & 0 & 18 & 50 & 29 & 15 & 0\\
NUM926+6 & 0.44 & 0 & 33 & 87 & 97 & 43 & 0\\
NUM926+7 & 0.47 & 4 & 28 & 80 & 77 & 37 & 0\\
NUN057+2 & 0.42 & 1 & 57 & 106 & 508 & 152 & 0\\
NUN060+1 & 0.27 & 0 & 18 & 14 & 33 & 17 & 0\\
NUN062+1 & 0.27 & 2 & 33 & 24 & 61 & 26 & 0\\
NUN066+2 & 0.22 & 0 & 111 & 194 & 3,061 & 3,110 & 55\\
NUN068+2 & 0.22 & 0 & 91 & 163 & 2,692 & 2,268 & 19\\
NUN072+2 & 0.44 & 3 & 92 & 188 & 961 & 276 & 0\\
\end{tabular}
\end{table}

\begin{table}
  \centering\small
  \caption{Conversion data for CASC-J11 problems (III/III).}
  \label{tab-conv-3}
\begin{tabular}{L{7em}R{2em}R{2em}R{4em}R{4em}R{5em}R{4em}R{2em}}
  \textbf{Problem} & \textbf{Rtg} & \textbf{T1} & \textbf{S1} &
  \textbf{S2} & \textbf{S3} & \textbf{S4} & \textbf{T2}\\\midrule
NUN076+1 & 0.27 & 0 & 30 & 24 & 61 & 31 & 0\\
NUN081+1 & 0.27 & 0 & 14 & 11 & 25 & 13 & 0\\
PRO002+3 & 0.47 & 3 & 145 & 294 & 4,291 & 329 & 0\\
PRO004+2 & 0.36 & 19 & 152 & 227 & 14,299 & 192 & 0\\
PRO009+1 & 0.28 & 3 & 194 & 287 & 4,655 & 491 & 1\\
PUZ001+2 & 0.31 & 0 & 43 & 100 & 173 & 62 & 0\\
PUZ078+1 & 0.25 & 1 & 200 & 238 & 51,359 & 507 & 4\\
PUZ133+1 & 0.44 & 0 & 112 & 177 & 1,505 & 468 & 1\\
REL025+1 & 0.50 & 3 & 295 & 740 & 85,513 & -- & --\\
REL050+1 & 0.29 & 1 & 158 & 373 & 30,213 & 2,219 & 235\\
RNG103+2 & 0.28 & 0 & 29 & 52 & 117 & 51 & 0\\
SET076+1 & 0.33 & 2 & 19 & 40 & 49 & 25 & 0\\
SET094+1 & 0.25 & 0 & 14 & 31 & 33 & 17 & 0\\
SET601+3 & 0.36 & 0 & 97 & 356 & 1,865 & 880 & 9\\
SET637+3 & 0.28 & 0 & 46 & 75 & 279 & 51 & 0\\
SET681+3 & 0.50 & 0 & 86 & 131 & 492 & 142 & 0\\
SET686+3 & 0.56 & 0 & 70 & 116 & 685 & 148 & 0\\
SEU187+2 & 0.36 & 9 & 86 & 179 & 729 & 179 & 0\\
SEU363+1 & 0.39 & 0 & 66 & 124 & 177 & 74 & 0\\
SWB013+2 & 0.28 & 44 & 57 & 74 & 129 & 52 & 0\\
SWC037+1 & 0.22 & 0 & 38 & 55 & 161 & 52 & 0\\
SWV161+1 & 0.33 & 2 & 25 & 58 & 77 & 30 & 0\\
SWV162+1 & 0.33 & 2 & 25 & 58 & 77 & 30 & 0\\
SWV167+1 & 0.25 & 3 & 29 & 54 & 109 & 36 & 0\\
SWV202+1 & 0.31 & 13 & 47 & 72 & 249 & 80 & 0\\
SWV203+1 & 0.31 & 13 & 47 & 72 & 249 & 80 & 0\\
SWV237+1 & 0.33 & 1 & 18 & 33 & 65 & 24 & 0\\
SWV401+1 & 0.31 & 1 & 47 & 81 & 289 & 114 & 0\\
SWV415+1 & 0.22 & 0 & 17 & 52 & 57 & 27 & 0\\
SWV417+1 & 0.25 & 0 & 17 & 49 & 37 & 13 & 0\\
SWV449+1 & 0.44 & 0 & 39 & 76 & 121 & 51 & 0\\
SWV451+1 & 0.44 & 16 & 77 & 165 & 1,122 & 176 & 0\\
SWV455+1 & 0.28 & 0 & 29 & 55 & 81 & 37 & 0\\
SWV457+1 & 0.42 & 0 & 91 & 164 & 1,829 & 132 & 0\\
SWV466+1 & 0.50 & 0 & 123 & 216 & 853 & 277 & 0\\
SWV470+1 & 0.67 & 0 & 77 & 143 & 497 & 74 & 0\\
SWV472+1 & 0.67 & 8 & 110 & 227 & 1,341 & 143 & 0\\
SWV477+1 & 0.42 & 25 & 104 & 177 & 3,393 & 307 & 0\\
SWV481+1 & 0.64 & 32 & 142 & 210 & 1,461 & 178 & 0\\
SWW229+1 & 0.31 & 2 & 18 & 72 & 73 & 21 & 0\\
SWW233+1 & 0.47 & 31 & 80 & 143 & 373 & 111 & 0\\
SWW473+1 & 0.25 & 14 & 16 & 81 & 57 & 19 & 0\\
SYN353+1 & 0.40 & 0 & 79 & 78 & 1,032 & -- & --\\
\end{tabular}
\end{table}

\begin{figure}
\begin{Verbatim}[fontsize=\footnotesize]
exp(Problem) :-
        getenv('PIE', PIE),
        format(atom(ScratchTPTP), '~w/scratch/scratch_tptp', [PIE]),
        consult(ScratchTPTP),
        get_time(T1),
        ppl_valid(tptp(Problem),
                  [mace=false,
                   proof=Proof,
                   prooftrans_options='expand renumber',
                   prooftrans_jterms,
                   r=true,
                   timeout=400]),
        get_time(T2),
        Proof = [proof(PSteps)],
        length(PSteps, L1),
        p9proof_install(Proof, [eq=binres_only]),
        current_p9proof(PureResolProof),
        length(PureResolProof, L2),
        get_time(T3),
        p9proof_to_ct(Proof, CT, [eq=binres_only]),
        get_time(T4),
        ct_tsize(CT, SCT),
        get_time(T5),
        ct_trafo_hyper(CT, CT1),
        get_time(T6),
        ct_tsize(CT1, SCT1),
        format('Problem: ~w~n', [Problem]),
        TP is T2-T1,
        format('Proving time: ~w~n', [TP]),
        format('Proof steps: ~w~n', [L1]),
        format('Pure binary resolution steps: ~w~n', [L2]),
        TC is T4-T3,
        format('Time for conversion to cut normal form tableau: ~w~n',
               [TC]),
        format('Tree size of cut normal form tableau: ~w~n', [SCT]),
        TH is T6-T5,
        format('Time for hyper conversion: ~w~n', [TH]),
        format('Tree size of hyper converted tableau: ~w~n', [SCT1]).
\end{Verbatim}
\caption{Code to run the described experiments in SWI-Prolog with \PIE. An
  example invocation would be \texttt{?- exp('ALG019-1').}}
\label{fig-exp-code}
\end{figure}

\closeout\plabelsfile
\closeout\plabelslogfile

\end{document}